\documentclass[
    reprint,
    amsmath,
    amssymb,
    aps,
    superscriptaddress
]{revtex4-1}

\usepackage{amsmath}
\usepackage{amsfonts}
\usepackage{amssymb}
\usepackage{mathtools}
\usepackage{bm}
\usepackage{braket}
\usepackage{bbold}
\usepackage{amsthm}

\usepackage{graphicx}
\usepackage[export]{adjustbox}
\usepackage{epstopdf}
\usepackage{float}
\usepackage{wrapfig}
\usepackage{sidecap}
\usepackage{afterpage}

\usepackage{dcolumn}
\usepackage{multirow}
\usepackage{diagbox}
\usepackage{cellspace}

\usepackage{subcaption}

\usepackage{algorithm}
\usepackage{algorithmicx}
\usepackage[noend]{algpseudocode}

\usepackage{tikz}
\usetikzlibrary{quantikz2}

\usepackage[english]{babel}
\usepackage{ragged2e}
\usepackage{verbatim}
\usepackage{enumitem}
\usepackage[normalem]{ulem}

\usepackage{aliascnt}

\newaliascnt{eqfloat}{equation}
\newfloat{eqfloat}{h}{eqflts}
\floatname{eqfloat}{Equation}

\theoremstyle{definition}

\newtheorem*{Pre*}{Proof}

\theoremstyle{plain}
\newtheorem{theorem}{Theorem}[section]

\newtheorem{lemma}[theorem]{Lemma}
\newtheorem{prop}[theorem]{Proposition}

\theoremstyle{definition}
\newtheorem{definition}{Definition}[section]

\usepackage{xcolor}

\definecolor{Cesar}{rgb}{0,0,0}
\newcommand{\CF}[1]{\textcolor{Cesar}{#1}}

\usepackage{hyperref}

\begin{document}


\title{Real-Space Chemistry on Quantum Computers:\\
A Fault-Tolerant Algorithm with Adaptive Grids and a Transcorrelated Extension}

\author{C\'esar Feniou} \email{Corresponding authors: cesar.feniou@gmail.com (CF)}
\affiliation{Sorbonne Université, LCT, UMR 7616 CNRS, 75005 Paris, France}
\affiliation{Qubit Pharmaceuticals, Advanced Research Department, 75014 Paris, France}

\author{Christopher Cherfan}
\affiliation{Qubit Pharmaceuticals, Advanced Research Department, 75014 Paris, France}
\affiliation{\'Ecole Polytechnique Fédérale de Lausanne (EPFL), 1015 Lausanne, Switzerland}

\author{Julien Zylberman}
\affiliation{Sorbonne Université, Observatoire de Paris, Université PSL, CNRS, LUX, 75005 Paris, France}

\author{Baptiste Claudon}
\affiliation{Sorbonne Universit\'e, LJLL, UMR 7198 CNRS, 75005 Paris, France}
\affiliation{Sorbonne Université, LCT, UMR 7616 CNRS, 75005 Paris, France}
\affiliation{Qubit Pharmaceuticals, Advanced Research Department, 75014 Paris, France}

\author{Jean-Philip Piquemal}\email{jean-philip.piquemal@sorbonne-universite.fr (JPP)}
\affiliation{Sorbonne Université, LCT, UMR 7616 CNRS, 75005 Paris, France}
\affiliation{Qubit Pharmaceuticals, Advanced Research Department, 75014 Paris, France}

\author{Emmanuel Giner}\email{emmanuel.giner@lct.jussieu.fr (EG)}
\affiliation{Sorbonne Université, LCT, UMR 7616 CNRS, 75005 Paris, France}


\begin{abstract}
First-quantized, real-space formulations of quantum chemistry on quantum computers are appealing: qubit count scales logarithmically with spatial resolution, and Coulomb operators achieve quadratic instead of quartic computational scaling of two‑electron interactions. However, existing schemes employ uniform discretizations, so the resolution required to capture electron–nuclear cusps in high-density regions oversamples low-density regions, wasting computational resources. We address this by deploying non-uniform, molecule-adaptive grids that concentrate points where electronic density is high. Using Voronoi partitions of these grids, the molecular Hamiltonian is expressed in a Hermitian form and in a transcorrelated, isospectral form that eliminates Coulomb singularities and yields cusp-free eigenfunctions. Both formulations slot naturally into quantum eigenvalue solvers: Hermitian Quantum Phase Estimation (QPE) and the recent generalised Quantum Eigenvalue Estimation (QEVE) protocol for its non-Hermitian, transcorrelated counterpart. Numerical validation on benchmark systems confirms this ab--initio framework paves a promising route to ground‑state chemistry on quantum hardware.
\end{abstract}

\maketitle


\section{Introduction}

\subsection{History and Motivation}
In 1982, Feynman suggested that quantum computers could efficiently simulate quantum systems~\cite{feynman2018simulating}, launching efforts to build quantum hardware and develop quantum algorithms with the potential to outperform classical methods for specific problems. In that regard, quantum chemistry stands out as a promising application where a \CF{genuine} quantum computing advantage could be found since available techniques such as the Quantum Phase Estimation (QPE) intends to provide polynomial cost solutions to the time-independent, non-relativistic Schrödinger equation. Of course, this comes at two conditions: i) to be able to provide an appropriate initial state; ii) to access a fault-tolerant quantum computer (FTQC) to handle the complex quantum circuits~\cite{kitaev1995quantum,whitfield2011simulation}. QPE for quantum chemistry was initially developed in the second-quantized framework, where the electronic Hamiltonian is represented on a set of \(N\) spinorbitals (\textit{e.g.} Gaussian atom-centered functions). The spectrum can be accessed either by time-evolving the system, originally via Trotterization~\cite{aspuru2005simulated, wecker2014gate, babbush2015chemical, low2023complexity}, later via Taylorization~\cite{berry2015simulating, babbush2018low}, or by applying a qubitized walk operator~\cite{low2019hamiltonian}. In this setting, the wavefunction requires \(\mathcal{O}(N)\) qubits (one per spinorbital), and the gate complexity becomes particularly favorable when the particle number $\eta$ is in the order of the orbital count~\cite{babbush2023quantum}. 

However, reaching the typical chemical accuracy, \textit{e.g.} an error of 0.04 eV compared to the exact total energy, requires reaching the complete basis-set (CBS) limit, which in turn necessitates using far more basis functions than particles ($\eta \ll N$)~\cite{traore2024shortcut}. 
In that context, the following first‐quantized formalism appears as a preferable scheme due to its low computational scaling in terms of both the number of particles and discretization points.
In this \textit{first‑quantized} picture, each of the \(\eta\) electrons has its own register that stores the binary index of a grid point among the \(N\) available sites.  
Indexing \(N\) sites requires only \(\lceil\log_2 N\rceil\) qubits, so the complete \(\eta\)-electron configuration space fits in \(\eta\lceil\log_2 N\rceil = \mathcal{O}(\eta\log N)\) qubits.  
Fermionic antisymmetry is enforced by projecting onto the subspace that is antisymmetric under exchange of these registers, as detailed in Ref.\,~\cite{su2021fault}. However, since standard overlapping integrals introduce \(\mathcal{O}(N^4)\) two-body terms, the gate complexity is driven to similarly high powers of \(N\), negating much of the practical advantage over second quantization. To address this, alternative discretization schemes have been developed with the goal of preserving a diagonal structure for the two-body operator in real space, thus reducing the Hamiltonian term count to \( \mathcal{O}(N^2) \). Examples include uniform real-space grids~\cite{kassal2008polynomial, kivlichan2017bounding, chan2023grid}, gausslets and mostly planewaves~\cite{childs2022quantum, babbush2019quantum, su2021fault, berry2024quantum, georges2025quantum}, and \CF{have} led to overall QPE gate complexities as small as \( \mathcal{O}(\eta^{8/3} N^{1/3}\epsilon^{-1} )\) using on‐the‐fly Hamiltonian term computation~\cite{babbush2019quantum, su2021fault}. Although asymptotically appealing, these discretisation schemes impose \textit{uniform spatial resolution}, which is particularly ill-suited to capture sharp features, known as cusps, that appear in the wavefunction near Coulomb singularities (e.g. where particles coincide)~\cite{troullier1991efficient}, since the fine resolution required near nuclei would simultaneously oversample low-density regions, leading to a large computational overhead. Moreover, accurate resolution of cusps inherently amplifies the spectral norms, since it explores near-singularity regions and thus correspondingly raises the total QPE complexity. 

\noindent
As a result, and despite new trade-offs and promising asymptotics, quantum computing still faces the age-old challenge of resolving electronic wavefunction cusps without incurring prohibitive computational cost~\cite{kivlichan2017bounding,mcclean2020discontinuous, traore2024shortcut}.

\noindent
In this work, we present a new framework for improved spatial accuracy in a real‐space, first‐quantized setting. The first strategy uses adaptive multicenter grids, rooted in Density Functional Theory (DFT) integration techniques~\cite{becke1988multicenter,mura1996improved}, to concentrate points in regions of high electron density, and implements the Laplacian via a Voronoi‐based finite‐volume scheme~\cite{sukumar2003voronoi,son2011voronoi,son2009theoretical}. We derive a Hermitian Hamiltonian in this basis and outline its integration into a qubitized-QPE protocol. The second strategy uses a similarity transformation to produce the \textit{transcorrelated} (TC) Hamiltonian~\cite{Boys_Handy_1969_determination}, which is isospectral to the original but replaces Coulomb singularities with finite effective interactions. The TC Hamiltonian eigenstates are free of electron–nucleus and electron–electron cusps, eliminating the need for excessively tight grids. On the drawback side, the TC Hamiltonian exhibits a three-body operator, and non-Hermitian terms, turning the Schrödinger eigenvalue problem into a non-self-adjoint (generalized) eigenvalue problem. Thus, standard QPE cannot be applied directly. \CF{Previous quantum-computing studies of transcorrelated Hamiltonians have mainly focused on quantum imaginary-time evolution and related non-Hermitian formulations in Gaussian basis sets~\cite{mcardle2020improving,sokolov2022orders,dobrautz2023ab,dobrautz2024toward}.} Recent advances in eigenvalue processing, however, have introduced Quantum Eigenvalue Estimation (QEVE) algorithms for non-Hermitian matrices~\cite{Low_2024}. Building on this idea, we derive the TC Hamiltonian in an adaptive-grid basis and outline how it can be integrated into a QEVE protocol.

\noindent
For each approach, we provide a detailed methodology and support its robustness with numerical simulations, while noting that the scheme offers numerous opportunities for improvement.

\subsection{Problem Statement}
This study addresses the central challenge of quantum chemistry: describing molecular electronic structure and properties through the non-relativistic, time-independent Schrödinger equation. This equation describes the $\eta$ electrons, located at positions $ \{\mathbf{x}_i\}_{i=0}^{\eta-1} \subset \mathbb{R}^3 $, which interact among themselves through a repulsive Coulomb potential and with the $ M $ nuclei with respective charges $ \{ \mathcal{Z}_\alpha \}_{\alpha=0}^{M-1} $ through an attractive Coulomb potential.  
As usual in quantum chemistry, we employ the Born-Oppenheimer approximation~\cite{BO} which leads to clamped nuclei, 
\textit{i.e.} the positions of the nuclei $ \{ \mathbf{R}_\alpha \}_{\alpha=0}^{M-1}$ are considered fixed.
Therefore, within the Born-Oppenheimer approximation, the total molecular Schrödinger equation is written as
\begin{equation}
\hat{H} \Psi(\mathbf{x}_0, \ldots, \mathbf{x}_{\eta-1}) = E \Psi(\mathbf{x}_0, \ldots, \mathbf{x}_{\eta-1}) \label{ps1},
\end{equation}
where the Hamiltonian $ \hat{H} $ is given, in atomic units, by
\begin{equation}
\hat{H} = \sum_{i=0}^{\eta-1} \Big{(} -\frac{1}{2}\nabla^{2}_{\mathbf{x}_i} 
  - \sum_{\alpha=0}^{M-1} 
  \frac{\mathcal{Z}_\alpha}{|\mathbf{x}_i - \mathbf{R}_\alpha|} 
  + \sum_{\substack{j>i}}^{\eta-1} 
  \frac{1}{|\mathbf{x}_i - \mathbf{x}_j|} \Big{)} \label{ps2}.
\end{equation}
In Eq. \eqref{ps1}, the wave function $ \Psi(\mathbf{x}_0, \ldots, \mathbf{x}_{\eta-1}) $ corresponds to the $i$-th eigenstate of the Hamiltonian, and as any fermionic wave function, is antisymmetric with respect to exchange of any couple of particles, \textit{i.e.}
\begin{equation}
\Psi(\ldots, \mathbf{x}_i, \ldots, \mathbf{x}_j, \ldots) = 
-\Psi(\ldots, \mathbf{x}_j, \ldots, \mathbf{x}_i, \ldots) \,\,\, \forall \,i,j\,,
\end{equation}
with $ |\Psi_i|^2 $ representing the probability density for the electron's positions in the $i$-th eigenstate. The real number $E_i$ in Eq. \eqref{ps1} is the total energy of the system associated to the $i$-th eigenstate, which is the central quantity needed to compute many physico-chemical properties in molecular systems. 
Among all eigenvalues of the Hamiltonian, the lowest one $E_0$ plays a special role, and can be recast as the solution of the following minimization problem
\begin{equation}
  E_{\text{GS}}^{*}=\min_{\Psi\in\mathcal{H}_{\mathrm{as}}^{1}(\mathbb{R}^{3\eta})}
  \frac{\displaystyle\int\Psi^{*}\hat{H}\Psi\;\prod_{i=0}^{\eta-1}d\mathbf{x}_i}
       {\displaystyle\int|\Psi|^{2}\;\prod_{i=0}^{\eta-1}d\mathbf{x}_i},
\end{equation}
where $ \mathcal{H}^1_\mathrm{as}(\mathbb{R}^{3\eta}) $ is the space of antisymmetric, square-integrable functions of $3\eta$ variables with finite kinetic energy. On classical computers, determining the energy of the ground state within chemical accuracy (\textit{i.e.} within an error of 0.0016 a.u. or equivalently 0.04 eV) exhibits exponential complexity with the size of the problem (space dimension and number of particles)~\cite{Helgaker2000}.

\subsection{First-Quantized Molecular Simulation via Qubitized QPE}
\label{LCUB}
This section reviews the qubitized‐QPE procedure for ground-state energy estimation of first‐quantized molecular Hamiltonians, 
and gives a qualitative and intuitive explanations of the leading order scaling complexities.
Consider the above electronic system of $\eta$ electrons, each described in a spatial discretization basis of $N$ elements. In the first-quantization formalism, each electron is encoded in a quantum register of $\log N$ qubits, and the $\eta-$electron wavefunction is obtained by tensor multiplying single-electron states. The number of qubits to encode this wave function scales therefore as $\mathcal{O}(\eta \log N)$. Since the fermionic antisymmetry is not embedded in the Hamiltonian, it must therefore be enforced directly on the wavefunction, a process that can be performed with negligible additional gate complexity $\mathcal{O}(\eta \log \eta \log N)$~\cite{Berry_2018}. 
The first quantized Hamiltonian, once discretized on a chosen spatial basis, is usually expressed as a linear combination of unitaries (LCU). From this decomposition one constructs the corresponding qubitized quantum-walk operator, unitary and spectrum-related to the original Hamiltonian, whose eigenvalues can then be extracted using the QPE algorithm~\cite{low2019hamiltonian}. This now usual protocol is outlined more precisely in Appendix \ref{SI-LCUB}. Phase estimation then uses $\widetilde{\mathcal O}\!\bigl(C_H\,\lambda\,\epsilon^{-1}\bigr)$
logical gates, where $\epsilon$ is the target energy error, $C_H$ is the cost of block-encoding the walk operator, $\lambda$ is the one‑norm of the LCU of the Hamiltonian matrix as defined in \ref{SI-LCUnorm}, and $\widetilde{\mathcal O}$ hides poly‑logarithmic factors~\cite{low2019hamiltonian}. $C_H$ has been shown to scale as $\tilde{\mathcal{O}}(\eta)$~\cite{babbush2018low}. The LCU norm is fundamentally governed by the \textit{minimal spatial separation} between electrons \footnote{as $O(\eta/(\Delta r)^2)$ for the kinetic part and $O(\eta^2/\Delta r)$ for the Coulomb part.}. For a uniform lattice or plane‑wave grid the spacing scales as $\Delta r\sim(\Omega/N)^{1/3}$ with volume $\Omega\propto\eta$.  Consequently, the one-norm of the LCU of the kinetic and potential operator are $\lambda_T=\mathcal O\!\bigl(\eta^{1/3}N^{2/3}\bigr)$ and $\lambda_V=\mathcal O\!\bigl(\eta^{5/3}N^{1/3}\bigr)$ respectively, so $\lambda=\mathcal O\!\bigl(\eta^{4/3}N^{2/3}+\eta^{8/3}N^{1/3}\bigr)$, yielding a total baseline scaling of $\widetilde{\mathcal O}\!\bigl((\eta^{4/3}N^{2/3}+\eta^{8/3}N^{1/3})\,\epsilon^{-1}\bigr)$
first given in ~\cite{su2021fault,babbush2018low}. \CF{Recently, Berry \textit{et al.} proposed a quantum fast-multipole method to evaluate the Coulomb potential with nearly linear dependence on \(\eta\) in first-quantized grid-based simulation, leading to a T-count \(\widetilde{\mathcal O}\!\left(\frac{\eta^{4/3}N^{2/3}+\eta^{1/3}N^{4/3}}{\epsilon}\right),\) and improving the scaling of real-space approaches for large particle counts~\cite{berry2025quantum}.
Conversely, for simulations requiring very large numbers of plane waves or grid points}, it has been proposed to compute the kinetic operator \textit{on the fly} in the momentum basis within the interaction picture, thus avoiding incurring the cost associated with the \(\lambda_T\) value of the kinetic operator as usually done using LCU. The QPE cost following this technique was shown to be reduced to \(\mathcal{O}\left(\eta^{8/3} N^{1/3} \, \epsilon^{-1}\right)\)~\cite{babbush2019quantum}.  Plane‐wave bases, though ideal for periodic solids, are poorly suited to chemistry and strongly correlated systems: their uniform spatial resolution forces massive oversampling of low‐density regions just to resolve core cusps and valence oscillations, leading to a significant computational overhead. \\
In practice, plane waves are always used in combination with pseudopotentials~\cite{troullier1991efficient}, which simplify the treatment of core regions by approximating core electrons as frozen and replacing their interactions with an effective external potential. Well aware of this, the embedding of pseudo-potentials to quantum simulation using planewaves has been proposed and quantified~\cite{berry2024quantum}. Nevertheless, the non locality of the pseudo potentials induces computational overheads or localization approximations, as usual in real-space calculations.
\section{Space Discretisation Tailored to Molecular Systems}
\label{sec:space_discr}
\subsection{Multi-center Molecular Grids}
In real-space grid methods, accurately resolving the Coulomb singularity requires local refinement near the nuclei. A natural strategy to increase point density in these regions is to construct a union of spherical grids centered around each nucleus, using a radial discretization that becomes progressively more concentrated near its center. Such discretizations have been extensively explored and optimized in the context of density functional theory (DFT), as they are commonly used for evaluating energy functionals via numerical integration over space. The radial discretization often follows a Becke scheme~\cite{Froese} which can be summarized as follows. 
Given a uniform partition $u_i$ of the interval $]0;1]$, the radial coordinate is discretized into $N_r$ points according to 
\begin{equation}
  r_{i}=-\alpha \ln \big( 1-u_{i}^\nu\big) \label{becke},
\end{equation}
where $\alpha$ and $\nu>0$ control both the overall range and the concentration of points around the origin. 
For the angular component, several options are available. One possibility is to use Gauss-Legendre quadratures~\cite{MATSUOKA1981387} such that the angular points are discretized into respectively $N_\theta$ and $N_\phi$ grid points $\theta_j \in [0;\pi]$ and $\phi_k\in [0;2\pi]$ via
\begin{equation}
  \begin{cases}
    \theta_j = \frac{\pi}{2}(x_{j}+1)\\
    \phi_{k} = \delta k,
  \end{cases} \label{gauss-leg}
\end{equation}
where $x_j$ is the $j$th root of the Legendre polynomial of degree $N_\theta$ (in ascending order), and the azimuthal angle repartition is uniform with constant step size $\delta$. Since the $\theta$ and $\phi$ partitions are built independently, every node has natural neighbors in $r$, $\theta$, and $\phi$, which simplifies finite--difference derivative approximations. However, as can be seen from Fig. \ref{fig:gauss}, were we illustrated the various schemes by representing the real-space grids obtained for the H$_2$O molecule, the Gauss-Legendre quadrature exhibits a dense clustering of points along lines of latitude close to the north and south poles of each sphere, thus breaking rotational symmetry. 
An alternative and potentially more suitable approach is to use the Lebedev angular quadrature, which is widely employed in multicenter numerical integration schemes, particularly in routine DFT calculations for molecular systems. It is based on a class of quadratures that are invariant around the octahedral point group, by solving non-linear equations that ensure invariant spherical harmonics~\cite{beentjesquadrature}. An illustrative example of a multicenter molecular grid combining Becke radial discretisation and Lebedev quadrature is shown in Fig. \ref{fig:lebedev}. From the latter plot, it clearly appears that, with respect to the grids based on the Gauss-Legendre angular discretization, the grids based on the Lebedev angular discretization are closer to the local rotational symmetry.
\begin{figure}[htp]
    \centering

    \begin{subfigure}[t]{0.48\linewidth}
        \centering
        \includegraphics[
            width=\linewidth,
            trim=100 200 100 200,
            clip
        ]{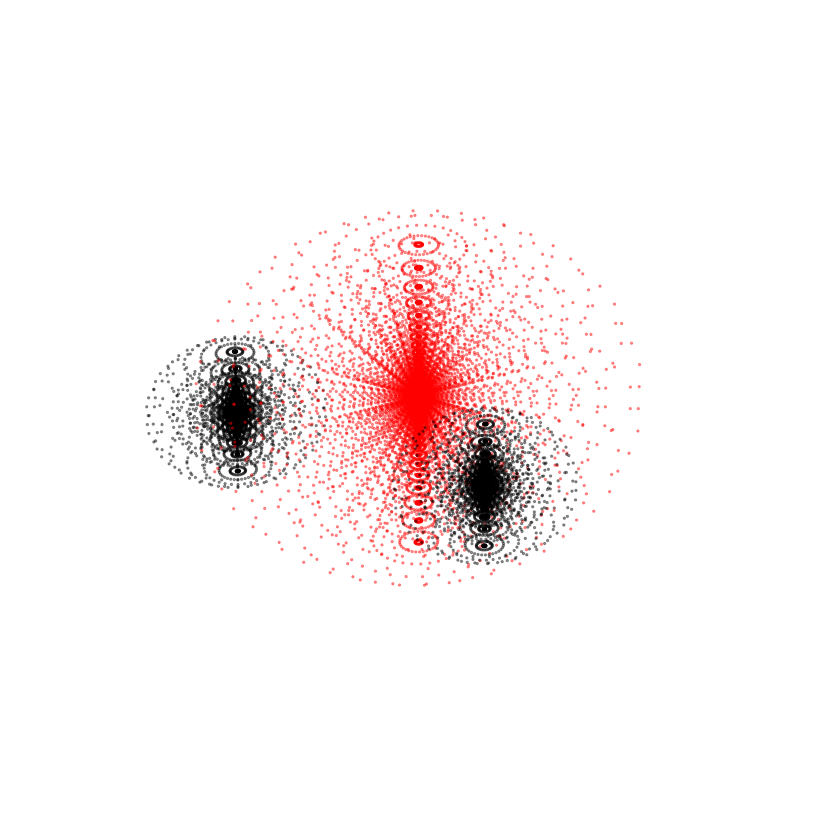}
        \caption{Gauss--Legendre}
        \label{fig:gauss}
    \end{subfigure}
    \hfill
    \begin{subfigure}[t]{0.48\linewidth}
        \centering
        \includegraphics[
            width=\linewidth,
            trim=100 200 100 200,
            clip
        ]{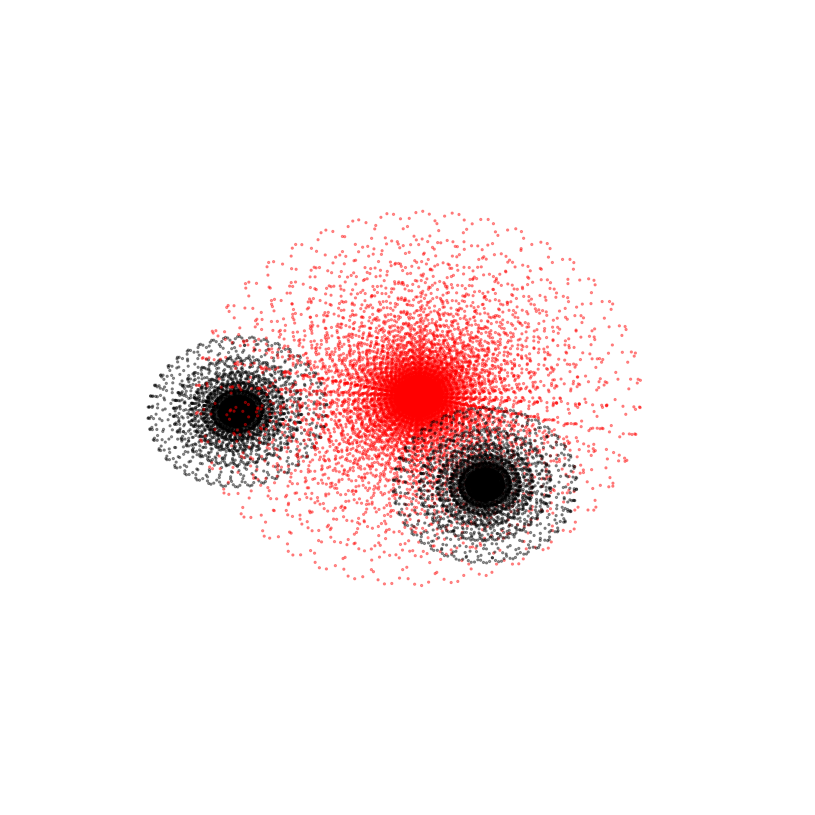}
        \caption{Lebedev}
        \label{fig:lebedev}
    \end{subfigure}

    \caption{Visualization, for a water molecule, of the two spatial discretizations used here, namely (a) Gauss--Legendre and (b) Lebedev. Atomic distances are not to scale, and colors are illustrative only.}
    \label{fig:grids}
\end{figure}
\vspace{-0.5cm}
\subsection{Voronoi Finite Volume Discretization}

Although a finite‐difference scheme might seem a natural choice on the grid described above, there are disadvantages in doing so: (i) the presence of multiple grids makes it ambiguous to define nearest neighbors for a given coordinate, especially if the grids are overlapping and (ii) there is no simple linear relation between the Laplacian in a grid centered at the origin \(\nabla^2_{\mathbf{x}_i}\), and the Laplacian on a grid shifted to \(\mathbf{R}_\alpha\), \textit{i.e.} \(\nabla^2_{\mathbf{x}_i-\mathbf{R}_\alpha}\).
To avoid these issues, and without imposing any restrictions on the underlying grid, we adopt a finite‐volume approach using Voronoi cells. This scheme has previously been used to classically solve the transport-diffusion equations~\cite{DU20033933,poveda2023}, and has been used by~\cite{sukumar2003voronoi} to solve the molecular Schrödinger equation using similar molecular multicenter grids. We briefly report here the main concepts used for our present approach.
\vspace{-1em}
\subsubsection{Voronoi Diagrams}
\begin{definition}[Voronoi Cells and Diagram]
 Given a set $\mathcal{V}$ of $N$ arbitrarily distributed points $\{\mathbf{r}_m\}_{m=0}^{N-1} \subset \mathbb{R}^{3} $, the Voronoi cell $\text{Vor}(\mathbf{r}_m)$ associated with point $\mathbf{r}_m$ is defined as the region of space closer to $\mathbf{r}_m$ than to any other point. In mathematical terms:
\begin{equation}
   \text{Vor}(\mathbf{r}_m) = \left\{ \mathbf{r} \in \mathbb{R}^3 : \|\mathbf{r} - \mathbf{r}_m\| \leq \|\mathbf{r} - \mathbf{r}_n\| \quad \forall \,\,m \neq n \right\}.
\end{equation}
The Voronoi diagram of $\mathcal{V}$ is then the union of all Voronoi cells in $\mathcal{V}$:
\begin{equation}
  \text{Vor}(\mathcal{V})=\bigcup_{m=0}^{N-1}\text{Vor}(\mathbf{r}_m).
\end{equation}
\end{definition}

\noindent
Note that the definition could be formulated for points in higher dimensions and for metrics other than the euclidean one. This partitioning yields a unique Voronoi diagram, in which the space is discretized into convex polyhedral cells~\cite{10.1145/116873.116880}.
\begin{definition}[Natural Neighbors and Voronoi Facets]
   Given a set $\mathcal{V}$ of $N$ arbitrarily distributed points $\{\mathbf{r}_m\}_{m=0}^{N-1} \subset \mathbb{R}^{3} $, a point $\mathbf{r}_n$ is called a natural neighbor of $\mathbf{r}_m$ if and only if 
   \begin{equation}
   \label{eq:gammamn}
     \Gamma_{mn}\coloneqq\text{Vor}(\mathbf{r}_m) \cap \text{Vor}(\mathbf{r}_n) \neq \emptyset.
   \end{equation}
   We denote by $\Lambda(m)$ the set of indices of the natural neighbors of $\mathbf{r}_m$, \textit{ie}
   \begin{equation}
     \Lambda(m)=\{n:\Gamma_{mn}\neq \emptyset\},
   \end{equation}
   and we call $\Gamma_{mn}$, the intersection of the Voronoi cells of two natural neighbors $\mathbf{r}_m$ and $\mathbf{r}_n$, a Voronoi facet.
\end{definition}
By construction, $\Gamma_{mn}$ lies on the perpendicular bisector of the segment connecting these two points. In 3D, its area $\sigma_{mn}$ can be calculated from its surrounding vertices $\{\mathbf{\xi}_k\}$ via
\begin{equation}
\sigma_{mn} = \frac{1}{2} \sum_{k=0}^{\Xi-1} \left\| (\mathbf{\xi}_k - \mathbf{\rho}_{mn}) \times (\mathbf{\xi}_{k+1} - \mathbf{\rho}_{mn}) \right\|,\label{eq:sigma_mn}
\end{equation}
where $\mathbf{\rho}_{mn} = (\mathbf{r}_m + \mathbf{r}_n)/2$, and $\Xi$ is the number of surrounding vertices, with $\mathbf{\xi}_{M+1} = \mathbf{\xi}_1$. The cell $ \text{Vor}(\mathbf{r}_m) $ has a volume $v_m$ that can be expressed as the sum of the volumes of pyramids with base of area $\sigma_{mn}$ and apex at $\mathbf{r}_m$:
\begin{equation}
v_m = \frac{1}{6} \sum_{n \in \Lambda(m)} |\mathbf{r}_{m}-\mathbf{r}_n| \sigma_{mn} \label{eq:vol}.
\end{equation}
There exists many algorithms that provide the Voronoi diagrams of a given set of points, \textit{ie} the location of the vertices of each Voronoi facet with their areas and the volume of each Voronoi cell~\cite{doi:10.1137/040617364,doi:10.1137/S0036144599352836}. Although their time complexity in worse-case scenarios could scale superpolynomially, in practice and for well-structured sets of points, the Voronoi diagram can be computed in polynomial time~\cite{10.1145/1137856.1137880}. As an illustration, we display in figure \ref{fig:vor} the Voronoi diagram generated by the \textit{Qhull} library~\cite{Qhull} for a set of points in 2D.
\begin{figure}[htp]
    \centering
    \begin{subfigure}[t]{0.45\linewidth}
        \centering
        \includegraphics[width=\linewidth]{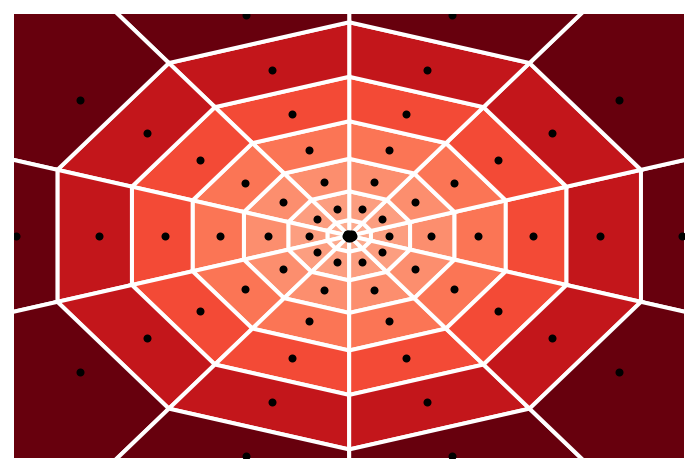}
        \caption{Single center}
        \label{fig:vor_single}
    \end{subfigure}
    \hfill
    \begin{subfigure}[t]{0.45\linewidth}
        \centering
        \includegraphics[width=\linewidth]{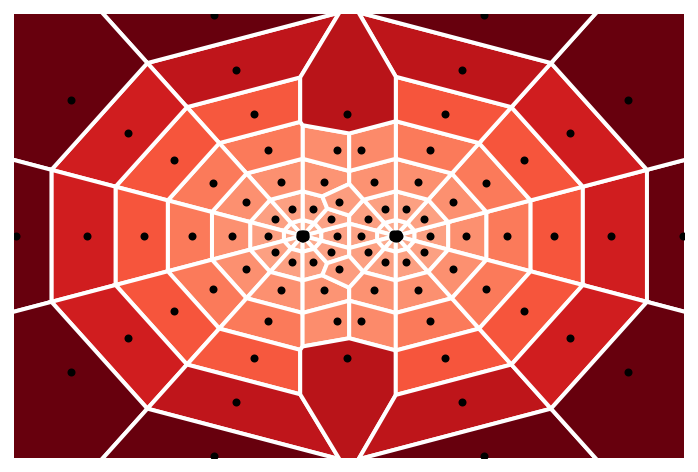}
        \caption{Bicenter}
        \label{fig:vor_bicenter}
    \end{subfigure}
    \caption{Voronoi diagrams for planar points following a Becke radial discretization with $\nu=1$, Eq.~\eqref{becke}, and a uniform angular partition: (a) single-center and (b) bicenter grids. The Voronoi cells associated with each black dot are delimited using shades of red proportional to their area.}
    \label{fig:vor}
\end{figure}
\vspace{-0.5cm}

\subsection{Wavefunction Representation}
Let us now address the problem of the representation of the wave function on the grids previously introduced.
Consider a single electron wavefunction discretized on a multi-center grid with $N$ points. It can be encoded in a $\log(N)$-qubit register as
\begin{equation}
  \ket{\psi_e} = \sum_{m=0}^{N-1}c_{m}\ket{m} \label{wavefun},
\end{equation}
where $m$ runs through the indices of the cells in the Voronoi diagram. The state $\ket{\psi_e}$ resides in the $N$-dimensional Hilbert space $\mathcal{H}_{e}$ and the expansion coefficients $c_m$ satisfy the usual orthonormality condition \(\sum_{m=0}^{N-1} |c_{m}|^{2} = 1\), provided that $\langle m'\ket{m} = \delta_{mm'}$.
For an $\eta$-electron system, the wavefunction is represented by a quantum register formed from the tensor product of $\eta$ single-electron wavefunctions. The total state vector thus lives in $\mathcal{H}_{\text{tot}} = \bigotimes_{i=1}^{\eta} \mathcal{H}_e$, and can be encoded using $\eta \log(N)$ qubits. Electrons being fermions, the electronic wavefunction must live in the antisymmetric subspace of $\mathcal{H}_{\text{tot}}$, which we ensure using the anti-symmetrization gate described in~\cite{Berry_2018}. 

\vspace{-1em}
\subsection{Hamiltonian representation} \label{finitevolumescheme}
We now give the matrix representation of the molecular Hamiltonian in the discretized space of the grid points. The mathematical derivation of the finite volume scheme is presented in Appendix \ref{SI-VFV}. The main idea is to integrate the molecular Schrödinger equation over the volumes of the Voronoi cells occupied by each electron, and take the average in the limit of infinitesimally small Voronoi cell volumes so as to express volume and surface integrals in terms of the parameters of the Voronoi diagrams. Scalar operators, like the Coulombic potential terms, are represented as diagonal matrices with each diagonal elements being the value of the operator evaluated for the corresponding grid point. Ultimately, finding the ground-state energy boils down to finding the smallest eigenvalue of the finite volume discretized Hamiltonian, which can be represented in matrix form as
\begin{equation}
\begin{split}
  \mathbf{{H}} &= \sum_{i=0}^{\eta-1}\Big{(}-\frac{1}{2}\sum_{m,n=0}^{N-1}\mathbf{{L}}_{mn}\ket{m}_{i}\bra{n}_{i} -\sum_{m=0}^{N-1}\mathbf{U}_{m}\ket{m}_{i}\bra{m}_{i}\Big{)}\\ &+\sum_{i=0}^{\eta-1}\sum_{j>i}^{\eta-1}\sum_{m,p=0}^{N-1}\mathbf{W}_{mp}\ket{m}_{i}\bra{m}_{i}\ket{p}_{j}\bra{p}_{j} ,
\end{split}  \label{eq:discrH}
\end{equation}
where the discretized Laplacian matrix has matrix elements
\begin{equation}
\mathbf{L}_{mn} =
\begin{cases}
\displaystyle -\frac{1}{v_m} \sum_{k \in \Lambda(m)} \frac{\sigma_{mk}}{|\mathbf{r}_{m}-\mathbf{r}_{k}|} & \text{if } m = n, \\[1ex]
\displaystyle \frac{1}{v_m} \frac{\sigma_{mn}}{|\mathbf{r}_{m}-\mathbf{r}_{n}|} & \text{if } n \in \Lambda(m), \\[1ex]
0 & \text{otherwise,} \label{laplacian}
\end{cases}
\end{equation}
with $v_m$ and $\sigma_{mn}$ defined in Eqs. \eqref{eq:sigma_mn} and \eqref{eq:vol}, 
the Coulombic electron-nuclei attraction operator is diagonal in the position basis, with the following elements:
\begin{equation}
\mathbf{U}_{m} =
\sum_{\alpha=0}^{M-1} \frac{\mathcal{Z}_\alpha}{|\mathbf{r}_m - \mathbf{R}_\alpha|} ,
\label{oneb}
\end{equation}
and the two-body electron-electron Coulomb repulsion operator is diagonal in the tensor-product position basis:
\begin{equation}
\mathbf{W}_{mp} =
\frac{1}{|\mathbf{r}_m - \mathbf{r}_p|} \label{twob}.
\end{equation}
The notation $\ket{m}_{i}\bra{n}_{i}$ refers to the outer product $\ket{m}\bra{n}$ in the space of the $i$th electron. More precisely:
\begin{equation}
\begin{split}
  \ket{m}_{i}\bra{n}_{i} &=\underbrace{\mathbf{I}_N\otimes\cdots\otimes \mathbf{I}_{N}}_{i\textup{-times}} \otimes \ket{m}\bra{n} \otimes \underbrace{\mathbf{I}_N\otimes\cdots\otimes \mathbf{I}_{N}}_{(\eta-i+1)\textup{-times}} \\
  &= \mathbf{I}_{N^{i}} \otimes \ket{m}\bra{n} \otimes \mathbf{I}_{N^{\eta-i-1}}
\end{split}
\end{equation} 
\CF{The representation of the molecular Hamiltonian on a multicenter grid is unusual in quantum chemistry, and it is useful to comment on convergence, sparsity, and Hermiticity. The use of Voronoi cell volumes should be interpreted as a finite-volume approximation: local quantities are averaged over cells, and differential operators are represented by fluxes through shared faces, so the corresponding quadrature error depends on the cell diameter and on the local variation of the function inside the cell. In molecular applications the precise error is problem dependent as it depends on the regularity of the wavefunction and the geometry of the adaptive grid. We therefore do not claim a universal convergence rate here, but rely on existing convergence results for Voronoi finite-volume discretizations of diffusion-type equations~\cite{DU20033933,poveda2023}.} Furthermore, in fixed-dimensional space, the Voronoi Laplacian is sparse, where the $m$-th row contains $|{\Lambda}(m)| + 1$ non-zero entries (one per neighbor and itself). This sparsity reflects the local connectivity of the Voronoi mesh, and does not scale with the total number of discretization points. Last but not least, due to the non-uniform grid spacing and the corresponding variation in Voronoi cell volumes, the discretized Laplacian does not have a Hermitian matrix representation. However, there is an easy way to symmetrize equation \eqref{laplacian} and re-obtain an Hermitian Hamiltonian, as it is done in Ref. \onlinecite{son2009theoretical}. In the latter scheme, the diagonal potential energy matrix remains intact while the modified Laplacian matrix, denoted by $\bar{\mathbf{L}}_{mn}$, has elements given by 
\begin{equation}
  \bar{\mathbf{L}}_{mn} = \sqrt{\frac{v_m}{v_n}}\mathbf{L}_{mn}. \label{bar_lmn}
\end{equation}
We can then rewrite equation \eqref{eq:discrH} in a more compact form:
\begin{equation}
\begin{split}
  \mathbf{\bar{H}} &= \sum_{m,n=0}^{N-1}\mathbf{\bar{T}}_{mn}\sum_{i=0}^{\eta-1}\ket{m}_{i}\bra{n}_{i} \\ &+ \sum_{m,p=0}^{N-1}\mathbf{W}_{mp}\sum_{j>i}^{\eta-1}\ket{m}_{i}\bra{m}_{i}\ket{p}_{j}\bra{p}_{j},
\end{split} \label{eq:hermH}
\end{equation}
where $\mathbf{W}_{mp}$ is defined in equation \eqref{twob} and 
\begin{equation}
  \mathbf{\bar{T}}_{mn} = -\frac{1}{2}\mathbf{\bar{L}}_{mn}-
\begin{cases}
\displaystyle \mathbf{U}_m &\text{if } m=n,\\
0 & \text{otherwise,}
\end{cases} \label{symmoneb}
\end{equation}
and where $\mathbf{\bar{L}}_{mn}$ are the symmetrized matrix elements defined in Eq. \eqref{bar_lmn}, and $\mathbf{U}_m$ defined in Eq. \eqref{oneb}.\\

\CF{
Equations \eqref{eq:hermH} and \eqref{symmoneb} complete the construction of the molecular Hamiltonian in the Voronoi-cell basis. The resulting representation is sparse, Hermitian, and expressed entirely in terms of one- and two-body matrix elements associated with the underlying Voronoi discretisation. For quantum processing, this matrix representation must subsequently be transcribed into quantum circuits suitable for quantum phase estimation (QPE). The corresponding block-encoding constructions and quantum algorithmic analysis are presented in Section~\ref{quantum_algorithmic_analysis}.}

\CF{Before turning to these quantum implementations, it is instructive to consider an alternative representation of the molecular Hamiltonian itself. Although the Voronoi discretisation efficiently captures the relevant spatial structure of the wave function, the short-range electron--nucleus and electron--electron cusp conditions remain encoded in the operator. A well-established strategy in electronic structure theory is to incorporate part of this short-range correlation directly at the operator level through a similarity transformation. In the following, we therefore consider the transcorrelated (TC) molecular Hamiltonian~\cite{Boys_Handy_1969_determination}.}

\section{Transcorrelated Molecular Hamiltonian}
\label{sec:TCHam}
This approach applies an isospectral similarity transformation by a so-called correlation factor, which when properly chosen replaces unbounded Coulomb potentials with non-diverging interactions, at the price of introducing non‐Hermitian differential terms together with a three-body interaction. 
The use of a TC Hamiltonian has multiple consequences which are important to highlight in the present context. Provided that the cusps conditions are included in the correlation factor itself, the effective interactions produced by the similarity transformation are non-diverging~\cite{Boys_Handy_1969_determination,nooijen1998elimination,10.1063/1.5116024,Giner_2021}, and therefore the TC eigenfunctions are free of cusps. 
This property is crucial in the present context as it prevents the need for extremely dense grid points near the nuclei, and therefore reduces the operator norm of the Laplacian. Also, as the TC wave function is relieved from the burden of the electron-electron cusp conditions, it is smoother near the electron-electron coalescence, which simplifies its description with a finite tensorized basis.
As mentioned above, the TC Hamiltonian exhibits a three-body interaction, thus increasing the complexity with the number of particles $\eta$, which nevertheless remains diagonal and bounded on the basis of the real-space grid and therefore does not increase the operator norm. 
A crucial aspect \CF{for quantum algorithms arises} from the non-Hermitian differential terms in the TC Hamiltonian, which prevent the use of qubitized QPE. To alleviate this problem, we instead propose to use the recently developed Quantum EigenValue Estimation (QEVE) algorithms, which efficiently handles non-Hermitian operators~\cite{Low_2024}. Appendix \ref{SI-QCPE} explicitly outlines the QEVE algorithm, which, like qubitized QPE, relies on block-encoding the non-Hermitian Hamiltonian. Below, we detail the full workflow: deriving the transcorrelated Hamiltonian and its discretized form on the adaptive grid, building the Pauli‐LCU decomposition for QEVE compatibility, and validating the approach with numerical simulations of helium atom and hydrogen molecules.

\subsection{Molecular Transcorrelated Hamiltonian in First Quantization}\label{nhermtqc}
In this section, we briefly introduce the transcorrelated molecular Hamiltonian and its main features. Given a symmetric function of the $\eta$-electron coordinates, labelled here $\tau(\{\mathbf{x}_i\})$, the transcorrelated Hamiltonian is then obtained by a similarity transformation of $\hat{H}$ by $e^{{\tau}}$, \textit{i.e.}
\begin{equation}
  \tilde{H}:=e^{-{\tau}} \hat{H} e^{{\tau}} = \hat{H}+\big{[}\hat{H},{\tau} \big{]} + \frac{1}{2}\big{[}[\hat{H},{\tau} ] ,{\tau}\big{]} .\label{sim}
\end{equation}
This equality is obtained from the Baker–Campbell–Hausdorff formula, which in the case of the TC Hamiltonian truncates at second-order, as the fundamental commutator $\big{[}\hat{H},{\tau} \big{]}$ involves a function and a second-order differential operator. This effective Hamiltonian introduces additional potential terms which modify the bare Coulomb interaction, and a non-hermitian first-order differential operator. 
Nonetheless, as any similarity transformation, it preserves the spectrum of the original operator, here the molecular Hamiltonian $\hat{H}$. 
Given parameters $\mu_{ne},\,\mu_{ee}\in \mathbb{R}^{+}$, and a function
\begin{equation}
  {\tau}(\{\mathbf{x}_i\}) = \sum_{i=0}^{\eta-1}\Big{(}\sum_{\alpha=0}^{M-1} g_{\mu_{ne}}(x_{i\alpha}) + \sum_{j>i}^{\eta-1} h_{\mu_{ee}}(x_{ij}) \Big{)}\label{tau}
\end{equation}
where $g_{\mu_{ne}}$ and $h_{\mu_{ee}}$ are respectively functions of the electron-nuclei and electron-electron distances $x_{i\alpha}=|\mathbf{x}_i - \mathbf{R}_\alpha|$ and $x_{ij}=|\mathbf{x}_i - \mathbf{x}_j|$, the similarity‑transformed (TC) Hamiltonian
\(\tilde H\) can be written as
\begin{equation}
\small
\begin{split}
  \tilde H
  =\hat H -&\sum_{i,\alpha}\mathcal{\tilde{K}}[g_{\mu_{ne}}](\mathbf{x}_i,\mathbf{R}_\alpha)
            -\sum_{i,\alpha<\beta}\mathcal{\tilde{B}}[g_{\mu_{ne}}](\mathbf{x}_i,\mathbf{R}_\alpha,\mathbf{R}_\beta)\\
        -& \sum_{i<j}\mathcal{\tilde{K}}[h_{\mu_{ee}}](\mathbf{x}_i,\mathbf{x}_j)
            -\sum_{i<j<k}\mathcal{\tilde{B}}[h_{\mu_{ee}}](\mathbf{x}_i,\mathbf{x}_j,\mathbf{x}_k)
\end{split}
\end{equation}

where the effective two‑ and three‑body electron-electron operators read
\begin{equation}
\small
\begin{split}
   \mathcal{\tilde{K}}[h_{\mu_{ee}}](\mathbf x_i,\mathbf x_j)&=
   \frac{1}{2}\Bigl[
      \nabla^{2}_{i}h+\nabla_{j}^2h
      +(\nabla_ih)^2+(\nabla_jh)^2
    \Bigr]\\&+\nabla_ih \!\cdot\!\nabla_i
  +\nabla_jh\!\cdot\!\nabla_j
\end{split}
\end{equation}

\begin{equation}
    \small
    \begin{split}
        \mathcal{\tilde{B}}[h_{\mu_{ee}}](\mathbf x_i,\mathbf x_j,\mathbf x_k) \;&=\;
  \nabla_i h_{ij}\!\cdot\!\nabla_i h_{ik}
  +\nabla_jh_{ji}\!\cdot\!\nabla_jh_{jk}
  \\&+\nabla_kh_{ki}\!\cdot\!\nabla_kh_{kj}
    \end{split}
\end{equation}
with $h \equiv h_{\mu_{ee}}$, \(h_{ij}=h(\mathbf{x}_i,\mathbf{x}_j)\) and \(\nabla_i\equiv\nabla_{\mathbf{x}_i}\).
\CF{The electron--nucleus terms can be similarly derived. Noting \(g_{i\alpha}=g_{\mu_{ne}}(\mathbf x_i,\mathbf R_\alpha)\), one gets
\begin{equation}
    \small
    \begin{split}
    &\mathcal{\tilde K}[g_{\mu_{ne}}](\mathbf x_i,\mathbf R_\alpha)\;=\;\frac12\left[\nabla_i^2 g_{i\alpha}+\left(\nabla_i g_{i\alpha}\right)^2\right]+\nabla_i g_{i\alpha}\cdot\nabla_i \\
    &\mathcal{\tilde B}[g_{\mu_{ne}}](\mathbf x_i,\mathbf R_\alpha,\mathbf R_\beta)\;=\;\nabla_i g_{i\alpha}\cdot\nabla_i g_{i\beta}.
    \end{split}
    \end{equation}}
An important feature is that, as in practical calculations the operator is discretized in a finite basis set, the original spectrum of $\hat{H}$ is only recovered in the limit of a complete basis set.  The advantage of the transcorrelation formalism relies therefore on the fact that if some physics is encoded in $\tau$, the convergence of the low-lying spectrum of the discretized transcorrelated Hamiltonian $\tilde{H}$ is faster than that of the discretized bare Hamiltonian $\hat{H}$. For instance, if the cusps conditions~\cite{kato1957eigenfunctions} are encoded within $\tau$, the right-eigenfunctions of $\tilde{H}$ are cusp-free, which suggests a faster convergence when expanded on a finite basis set. There \CF{are} therefore a wide variety of functional forms for this correlation factor $\tau$ all built in terms of the so-called one-, two- and three-body functions which depend on the electron-nuclei, electron-electron and electron-electron-nuclei coordinates, respectively. In this work, we choose the functions $g_{\mu_{ne}}$ and $h_{\mu_{ee}}$ in equation \eqref{tau} to be 
\begin{equation}
  g_{\mu_{ne}}(x_{i\alpha})=x_{i\alpha}\big{(}\text{erf}(\mu_{ne} x_{i\alpha})-Z_\alpha \big{)} + \frac{1}{\mu\sqrt{\pi}}e^{-(\mu_{ne} x_{i\alpha})^2} \label{g},
\end{equation}
\begin{equation}
  h_{\mu_{ee}}(x_{ij})= x_{ij}\big{(}1 - \text{erf}(\mu_{ee} x_{ij})\big{)} - \frac{1}{\mu_{ee}\sqrt{\pi}}e^{-(\mu_{ee} x_{ij})^2} \label{h}.
\end{equation}
as it is done in the work of~\cite{Giner_2021}. The form of the functions in Eqs. \eqref{g} and \eqref{h} are such that they reproduce, at leading order in $1/x$, an effective smooth interaction instead of the bare diverging Coulombic interaction (see Ref. \onlinecite{Giner_2021} for more details). As the parameters $\mu_{ne}$ and $\mu_{ee}$ are decreased, the effect of the correlation factor $\tau$ becomes more important.

\noindent
With the choices made for $g_{\mu_{ne}}$ and $h_{\mu_{ee}}$ in equations \eqref{g} and \eqref{h}, and the expressions derived in Ref.~\cite{Giner_2021}, we can proceed as previously and express the transcorrelated Hamiltonian in the basis of the grid points using the Voronoi finite volume scheme:
\begin{widetext}

\begin{equation}
\begin{split}
  \mathbf{\tilde{H}} &= \sum_{i=0}^{\eta-1} \Big{(}\sum_{m,n=0}^{N-1}(-\frac{1}{2}\mathbf{L}_{mn}-\mathbf{\tilde{D}}_{mn}^{(ne)})\ket{m}_{i}\bra{n}_{i} - \sum_{m=0}^{N-1}\mathbf{\tilde{U}}_{m}\ket{m}_{i}\bra{m}_{i} - \sum_{j>i}^{\eta-1}\sum_{m,p=0}^{N-1} \mathbf{\tilde{W}}_{mp}\ket{m}_{i}\bra{m}_{i}\ket{p}_{j}\bra{p}_j \\
  &- \sum_{j>i}^{\eta-1} \;\sum_{m,n,p,q=0}^{N-1} \mathbf{\tilde{D}}_{mnpq}^{(ee)} \ket{m}_{i}\bra{n}_{i}\ket{p}_{j}\bra{q}_{j} - \sum_{j>i}^{\eta-1}\sum_{k>j>i}^{\eta-1} \;\sum_{m,p,t=0}^{N-1} \mathbf{\tilde{{B}}}_{mpt}\ket{m}_{i}\bra{m}_{i}\ket{p}_{j}\bra{p}_{j}\ket{t}_{k}\bra{t}_{k}\Big{)}
\end{split}
\end{equation}
where 
\begin{equation}
\mathbf{\tilde{D}}_{mn}^{(ne)}=
\begin{cases}
\displaystyle \sum_{\alpha=0}^{M-1}\frac{-\mathcal{Z}_{\alpha}+\text{erf}(\mu r_{m\alpha})}{2 v_m} \sigma_{mn}\ \mathbf{\hat{r}}_{mn} \cdot \hat{\mathbf{r}}_{m\alpha} & \text{if } n \in \Lambda(m)\\[1ex]
0 & \text{otherwise}
\end{cases}  \label{nh1}
\end{equation}

\begin{equation}
\mathbf{\tilde{D}}_{mnpq}^{(ee)} =
\begin{cases}
\displaystyle \frac{1-\text{erf}(\mu r_{mp})}{2}\Big{(} \frac{\sigma_{mn}}{v_m}\mathbf{\hat{r}}_{mn} - \frac{\sigma_{pq}}{v_p} \mathbf{\hat{r}}_{pq} \Big{)} \cdot \hat{\mathbf{r}}_{mp} & \text{if } n \in \Lambda(m) \text{ and } q \in \Lambda(p) \\[1ex]
0& \text{otherwise}
\end{cases} \label{nh2}
\end{equation}

\begin{equation}
\begin{split}
\mathbf{\tilde{U}}_{m} =&
\sum_{\alpha=0}^{M-1}\Big{(}\frac{-\mathcal{Z}_{\alpha}+\text{erf}(\mu r_{m\alpha})}{r_{m\alpha}} - \frac{\mu}{\sqrt{\pi}}e^{-(\mu r_{m\alpha})^2}+ \frac{(\mathcal{Z}_{\alpha}-\text{erf}(\mu r_{m\alpha}))^2}{2} \Big{)} \\
&+ \sum_{\alpha=0}^{M-1}\sum_{\beta=0}^{M-1}\Big{(}\big{(}\mathcal{Z}_{\alpha}-\text{erf}(\mu r_{m\alpha})).(\mathcal{Z}_{\beta}-\text{erf}(\mu r_{m\beta})\big{)} \,\hat{\mathbf{r}}_{m\alpha}. \hat{\mathbf{r}}_{m\beta} \Big{)}
\end{split}
\end{equation}

\begin{equation}
  \mathbf{\tilde{{W}}}_{mp} = \frac{1-\text{erf}(\mu r_{mp})}{r_{mp}} + \frac{\mu}{\sqrt{\pi}}e^{-(\mu r_{mp})^2}+ \frac{(1-\text{erf}(\mu r_{mp}))^2}{2}
\end{equation}

\begin{equation}
  \mathbf{\tilde{{B}}}_{mpt}=\big{(}1-\text{erf}(\mu r_{mp})).(1-\text{erf}(\mu r_{mt})\big{)} \,\hat{\mathbf{r}}_{mp}. \hat{\mathbf{r}}_{mt} \label{three}
\end{equation}
\end{widetext}

which can be rewritten as
\begin{equation}
\begin{split}
  \mathbf{\tilde{H}} &= \sum_{m,n=0}^{N-1}\mathbf{\tilde{T}}_{mn} \sum_{i=0}^{\eta-1}\ket{m}_{i}\bra{n}_i \\&+\frac{1}{2}\sum_{m,n,p,q=0}^{N-1}\mathbf{\tilde{W}}_{mnpq}\sum_{i\neq j}^{\eta-1}\ket{m}_{i}\bra{n}_{i}\ket{p}_{j}\bra{q}_{j}\\&-\frac{1}{3}\sum_{m,p,t=0}^{N-1}\mathbf{\tilde{B}}_{mpt}\sum_{i\neq j\neq k}^{\eta-1}\ket{m}_{i}\bra{m}_{i}\ket{p}_{j}\bra{p}_{j}\ket{t}_{k}\bra{t}_{k}
\end{split} \label{eq:transH}
\end{equation}where $\mathbf{\tilde{B}}_{mpt}$ is defined in \eqref{three} and 
\begin{equation}
  \mathbf{\tilde{T}}_{mn}=-\frac{1}{2}\mathbf{L}_{mn} - \mathbf{\tilde{D}}_{mn}^{(ne)}-\begin{cases}
\displaystyle \mathbf{\tilde{U}}_m &\text{if } m=n\\
0 & \text{otherwise.}
\end{cases} 
\end{equation}
\begin{equation}
\mathbf{\tilde{W}}_{mnpq} = -\mathbf{\tilde{D}}_{mnpq}^{(ee)}-
\begin{cases}
\displaystyle \mathbf{\tilde{W}}_{mp} &\text{if } m=n \text{ and } p=q\\
0& \text{otherwise}
\end{cases} 
\end{equation}

\CF{The developments above complete the construction of the transcorrelated molecular Hamiltonian in the Voronoi-cell basis, resulting in a sparse matrix representation similar to that of the Hermitian formulation. However, the transcorrelated Hamiltonian is non-Hermitian: although it remains isospectral to the original molecular Hamiltonian, standard QPE cannot be applied directly, and molecular energies must instead be recovered using generalised-eigenvalue-processing methods introduced in Section~\ref{quantum_algorithmic_analysis}.}

\CF{The following section outlines the quantum workflow considered for the Hamiltonians constructed above, from their circuit-level encoding to eigenvalue-processing procedures for estimating the ground-state energy, together with the associated complexity analysis and discussion.}

\section{Quantum Algorithmic Analysis}
\label{quantum_algorithmic_analysis}
\subsection{Quantum Eigenvalue Processing}
\CF{As reviewed in Sec.~\ref{LCUB}, for a Hermitian Hamiltonian, the block-encoding is converted by qubitization into a unitary walk operator whose eigenphases encode the eigenvalues, and QPE on this walk gives the cost \(\widetilde O(C_H\cdot\lambda/\epsilon)\). This strategy is not directly applicable to the transcorrelated Hamiltonian because of its non-Hermitian nature. The generalized Quantum Eigenvalue Estimation (QEVE) procedure introduced by Low \textit{et al.}, and reviewed in Appendix~\ref{SI-QCPE}, replaces the unitary-walk picture of qubitized QPE by a direct Chebyshev-polynomial eigenvalue transformation of the block-encoded non-Hermitian matrix. The resulting query complexity becomes
\begin{equation}
    \widetilde O\!\left(C_H\lambda \cdot\alpha_U/\epsilon\right),\qquad\alpha_U \ge\max_{\ell}\left\|U_\ell\!\left(\mathbf H/\lambda\right)\right\|
\end{equation}
where \(U_\ell(x)\) denotes the \(\ell\)-th Chebyshev polynomial of the second kind. The factor \(\alpha_U\) quantifies the possible norm growth of these Chebyshev transformations when applied to the non-Hermitian operator.
In either the QPE and QEVE case, the block-encoding cost \(C_H\) and the normalisation \(\lambda\) are key cost drivers. The remainder of this section is therefore devoted to estimating these quantities for the Hamiltonians considered in this work.}
\subsection{Block-Encoding of the Hamiltonian}
\CF{We consider two complementary strategies for constructing the Hamiltonian block-encoding. The first is a Pauli-LCU approach, in which the Hamiltonian is assembled and decomposed classically before loading to the quantum computer. The second, discussed in the following subsection, is an arithmetic approach, in which the Hamiltonian coefficients are computed on-the-fly on the quantum computer.}

\subsubsection{Pauli LCU Approach}
\CF{The Pauli-LCU strategy relies on classical preprocessing: the real-space Hamiltonian matrix is first assembled from the grid and Voronoi data, and then decomposed into Pauli strings using the fast Walsh--Hadamard transform of Georges \textit{et al.}~\cite{Georges_2025,georges2025quantum}. This gives a classical coefficient table
\begin{equation}
\mathbf H=\sum_j a_j P_j 
\end{equation}
where \(P_j\) is a Pauli string acting on the first-quantized electron registers.
The corresponding block-encoding is then built in the usual LCU form: \texttt{PREP} coherently loads amplitudes proportional to \(\sqrt{|a_j|}\), and \texttt{SEL} applies the Pauli string indexed by \(j\). The explicit Pauli decompositions for the standard and transcorrelated Hamiltonian are given in Appendix~\ref{SI-LCUB}.}

\CF{For the standard real-space Hamiltonian, the number of Pauli coefficients inherits the \(O(N^2)\) size of the one- and two-electron data. In the transcorrelated case, the non-diagonal two-body term and the additional three-body contribution further enlarge this set, it can be reduced in practice~\cite{christlmaier2023xtc}. The Pauli-LCU approach is therefore simple at the circuit level and well suited to small systems, but its reliance on polynomial-in-\(N\) classical preprocessing and coefficient loading limits its scalability. The next subsection considers an alternative with larger circuit overhead but improved scaling in $N$.}

\subsubsection{Arithmetic approach}
\CF{The Hamiltonian entry oracle is constructed as a sequence of reversible arithmetic routines. The only classical data loaded consist of the nuclear charges and positions \(\{\mathbf{R}_\alpha , \mathcal{Z}_\alpha \}_{\alpha=0}^{M-1}\), the grid parameters \((N_r,N_{\theta})\), and the constants defining the radial and angular maps. Given grid indices $m$ and $n$, the computation proceeds in three stages:}
\begin{enumerate}[label=(\roman*)]
    \item \textit{Coordinate synthesis} -- reconstructing the real-space coordinates of a grid point:
    \begin{equation}
    \ket{m}\ket{0}\mapsto\ket{m}\ket{\mathbf r_m}.
    \end{equation}

    \item \textit{Voronoi synthesis} -- computing the geometric quantities defining the Voronoi discretization:
    \begin{equation}
        \ket{\mathbf r_m}\ket{\mathbf r_n}\ket{0}
        \mapsto
        \ket{\mathbf r_m}\ket{\mathbf r_n}\ket{d_{mn},v_m,v_n,\Gamma_{mn},\chi_{mn}}
    \end{equation}

    \item \textit{Hamiltonian synthesis} -- evaluating the one- and two-electron matrix elements from the synthesized geometric and molecular data:
    \begin{equation}
    \begin{aligned}
    \ket{d_{mn},v_m,\sigma_{mn},\Gamma_{mn}}\ket{0}
    &\mapsto
    \ket{d_{mn},v_m,\sigma_{mn},\Gamma_{mn}}\ket{h_{mn}},\\[0.3ex]
    \ket{d_{mp}}\ket{0}
    &\mapsto
    \ket{d_{mp}}\ket{W_{mp}}.
    \end{aligned}
    \end{equation}
\end{enumerate}

\CF{The Voronoi synthesis stage (ii) provides the geometric quantities entering the finite-volume discretisation, from which the kinetic matrix elements \(L_{mn}\), Eq.~(\ref{laplacian}), the nuclear attraction term \(U_m\), Eq.~(\ref{oneb}), and the electron-electron interaction kernel \(W_{mp}\), Eq.~(\ref{twob}), are evaluated.}

\CF{All arithmetic is performed with $b$ bits of fixed-point precision, with \(b=O(\log(1/\epsilon))\), and we denote the total number of grid points by \(N=M N_r N_{\theta}\). The oracle uses \(O(b\log N)\) ancillary qubits to store indices and fixed-point registers. Its gate complexity is dominated by the Voronoi synthesis step. In a direct implementation, the Voronoi cell associated with a grid point is obtained by testing the \(N-1\) other grid points as potential neighbors, giving a per-query cost \(\widetilde O(Nb^2)\). Here the factor \(N\) comes from the neighbor search, while the factor \(b^2\) accounts for fixed-point arithmetic such as coordinate reconstruction, distance evaluation, and geometric comparisons. If the grid structure makes it possible to identify a local set of candidate Voronoi neighbors of bounded size \(d_{\max}\), the Voronoi contribution can, in principle, be reduced to \(\widetilde O(d_{\max}b^2)\).}

\CF{The resulting arithmetic routines provide coherent access to the Hamiltonian coefficients: the one-electron coefficient \(h_{mn}=-L_{mn}/2-U_m\delta_{mn}\) and the two-electron Coulomb coefficient \(W_{mp}\). These routines define the value oracle for the usual sparse-access Hamiltonian~\cite{low2019hamiltonian}: after a coefficient is computed into a fixed-point register, it controls an ancilla rotation that loads an amplitude proportional to its magnitude, and the work registers are then uncomputed. Together with the corresponding sparse-index oracle specifying the addressed basis operator, \(\ket{m}_i\bra{n}_i\) for the one-electron part and \(\ket{m}_i\bra{m}_i\ket{p}_j\bra{p}_j\) for the two-electron part, this provides the block-encoding primitives for \(\mathbf H\). Thus the Hamiltonian is accessed by coherent on-the-fly coefficient synthesis from the compact molecular and grid data, rather than by loading precomputed matrix elements.}

\subsection{Norm scaling}
\CF{As reviewed in Sec.~\ref{LCUB}, the block-encoding normalisation is controlled by the smallest spatial scale resolved by the discretisation. This is particularly relevant for adaptive real-space grids, which deliberately cluster points near nuclei and cusp regions. The resulting minimum spacing \(h_{\min}\) can decrease much faster with \(N\) than on a uniform grid, leading to larger local kinetic and Coulomb matrix elements and therefore to a larger \(\lambda\).
For example, the Becke-type map
\begin{equation}
r_i=-\alpha\ln\!\left(1-u_i^\nu\right)
\end{equation}
transforms a uniform grid in \(u\) into a grid concentrated near the nucleus, where \(r_i\sim \alpha u_i^\nu\) and thus \(h_{\min}=O(N^{-\nu})\). The same mechanism that improves cusp resolution therefore increases the local operator norm. This increase is not, by itself, evidence against adaptivity. A uniform grid capable of resolving the same local physics must ultimately reach a comparable local resolution, typically at the cost of many more grid points. The relevant comparison is therefore the total cost required to achieve a target physical accuracy rather than the normalisation at fixed \(N\). This viewpoint is reinforced by LCU norm lower bounds introduced by Loaiza \textit{et al.}~\cite{loaiza2023reducing}: for a fixed represented Hamiltonian, the optimal normalisation cannot be arbitrarily smaller than a quantity proportional to the spectral width, specifically \((E_{\max}-E_{\min})/2\) in the Hermitian case. Thus part of the growth of \(\lambda\) reflects genuine physical resolution requirements rather than an artifact of the discretisation. Adaptive grids trade larger local matrix elements for smaller basis sizes. Assessing whether this tradeoff is favorable requires system-specific studies.}

\CF{In the transcorrelated case, the normalisation receives an additional contribution from the induced three-body interaction, giving
\begin{equation}
\lambda_{\rm TC}=\mathcal O\!\left(\frac{\eta}{h_{\min}^2}+\frac{\eta^2}{h_{\min}}+\eta^3\right)
\end{equation} For the small molecular systems targeted by very fine real-space discretisations, the three-body contribution is not expected to dominate over the kinetic term when \(h_{\min}\) is small. Moreover, the induced three-body term can be approximated with an effective \(\eta^2\)-scaling contribution, as proposed in Ref.~\cite{christlmaier2023xtc}.}

\vspace{-1em}
\section{Computational details}
This section outlines the numerical setup for the simulations presented in the next section. Voronoi diagrams were generated using the Qhull library~\cite{Qhull}. For one‐electron Hamiltonians, whose matrices comprise only a few thousand entries, we obtained exact spectra via direct diagonalization with standard linear‐algebra routines. Two‐electron systems square the Hilbert‐space dimension, so we employed the parallel, non‐Hermitian Davidson solver~\cite{doi:10.1137/0915004} in Quantum Package 2.0~\cite{qp2} to compute ground‐state energies and wavefunctions. The initial guess in the Davidson solver was a cc-PVDZ Hartree–Fock solution evaluated on the grid, \CF{corresponding to the Hartree--Fock input state proposed for the first-quantized QPE and QEVE eigenvalue-estimation procedures discussed in Appendix~\ref{SI-initial_state}.}

\section{Simulations}
In this section, we present the results of numerical simulations to support our Voronoi finite volume scheme and its transcorrelated extension. We first discuss the impact of the correlation factor on the discretized wave functions expressed on our real-space grids, and then report some results on energy differences.

\noindent 
Let us begin by investigating the impact of the nucleus-electron correlation factor on the wave function in the case of one-electron systems. 
We report in Figure \ref{fig:combined_scatter} the plots of the ground-state wave functions obtained for the H atom and $\text{H}_2^+$ system at 1 Å, with a uniform angular and Becke radial discretization ($\nu=1$) for 3000 total points, and for various choices of the $\mu_{ne}$ parameter tuning the one-body part of $\tau$. As can be seen from Figure \ref{fig:combined_scatter}, the function $g_{\mu_{ne}}$ smooths out the sharp nuclear cusps, and decreasing the parameter $\mu_{ne}$ progressively broadens the central peaks, reducing the need for excessively dense grids at the cusps.

\begin{figure}[h]
 \centering
 \setlength{\tabcolsep}{2pt} 
 \renewcommand{\arraystretch}{0.5} 
 \captionsetup[subfigure]{labelformat=empty}

 \begin{tabular}{ccc}
  \subcaptionbox{Non-TC, $\mu_{ne}\!=\!\!\infty\!$\label{fig:H_std}}
   [0.32\linewidth]{%
    \includegraphics[width=\linewidth,trim=120 100 120 119,clip]{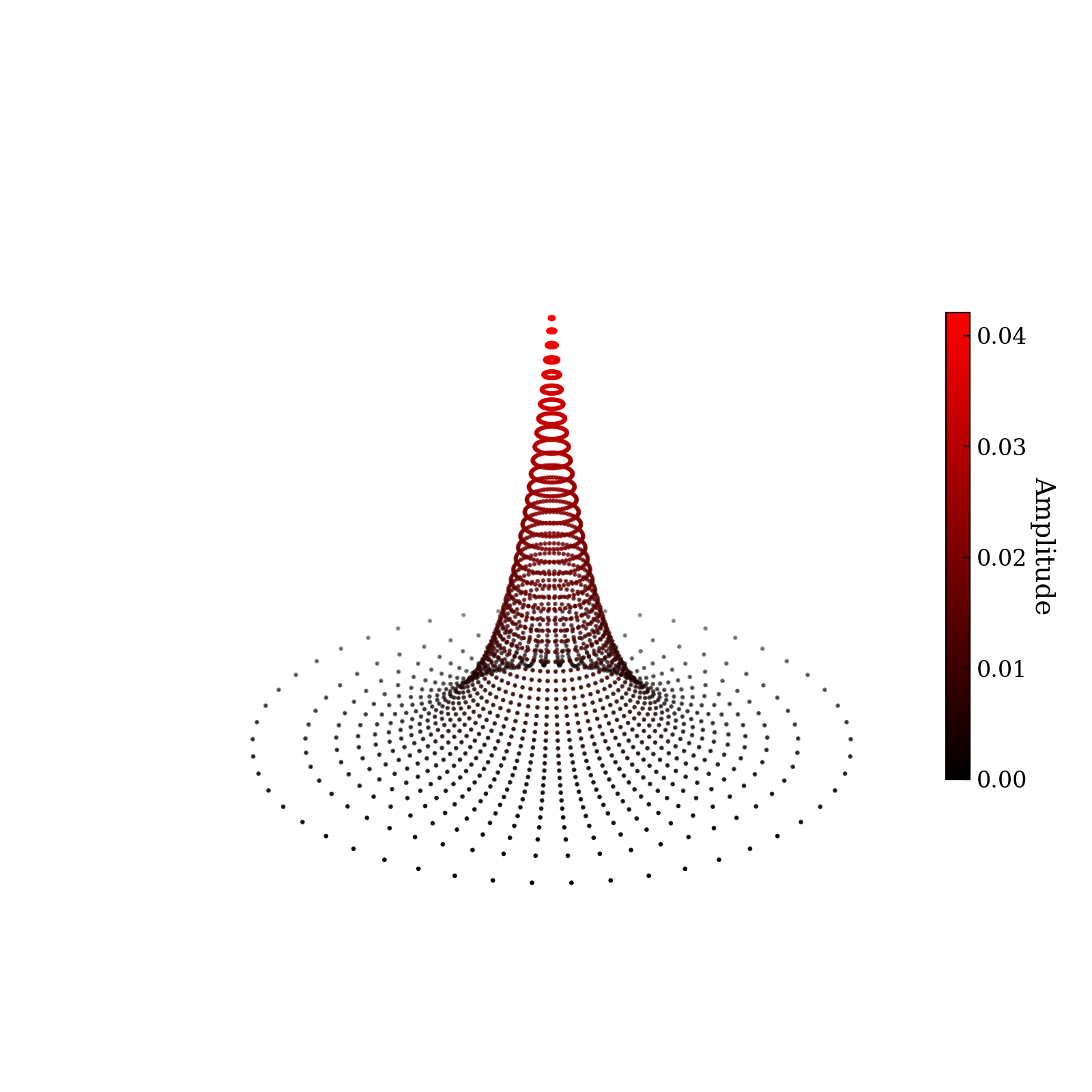}}
  &
  \subcaptionbox{$\mu_{ne}=3$\label{fig:H_mu3}}
   [0.32\linewidth]{%
    \includegraphics[width=\linewidth,trim=120 100 120 119,clip]{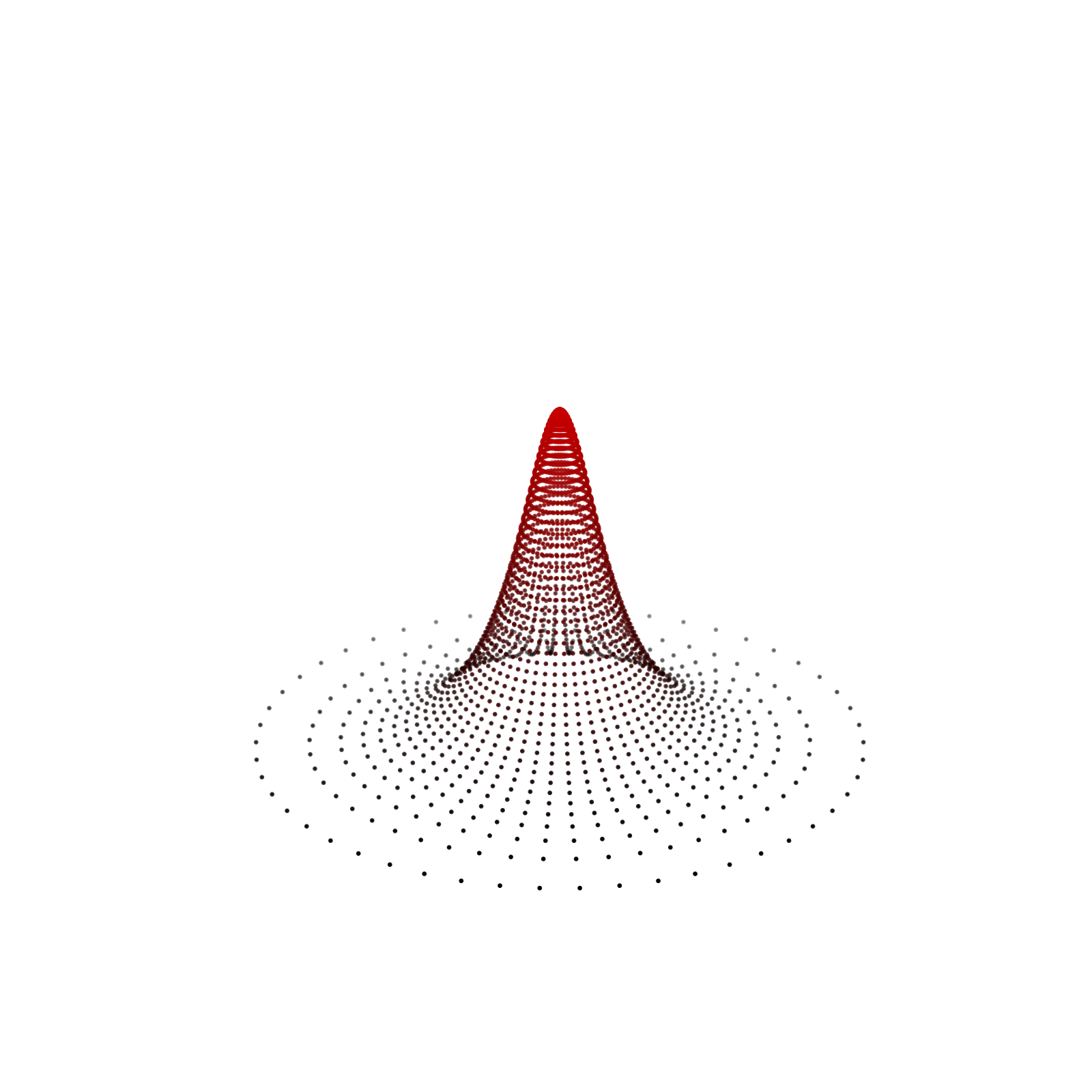}}
  &
  \subcaptionbox{$\mu_{ne}=1$\label{fig:H_mu1}}
   [0.32\linewidth]{%
    \includegraphics[width=\linewidth,trim=120 100 120 119,clip]{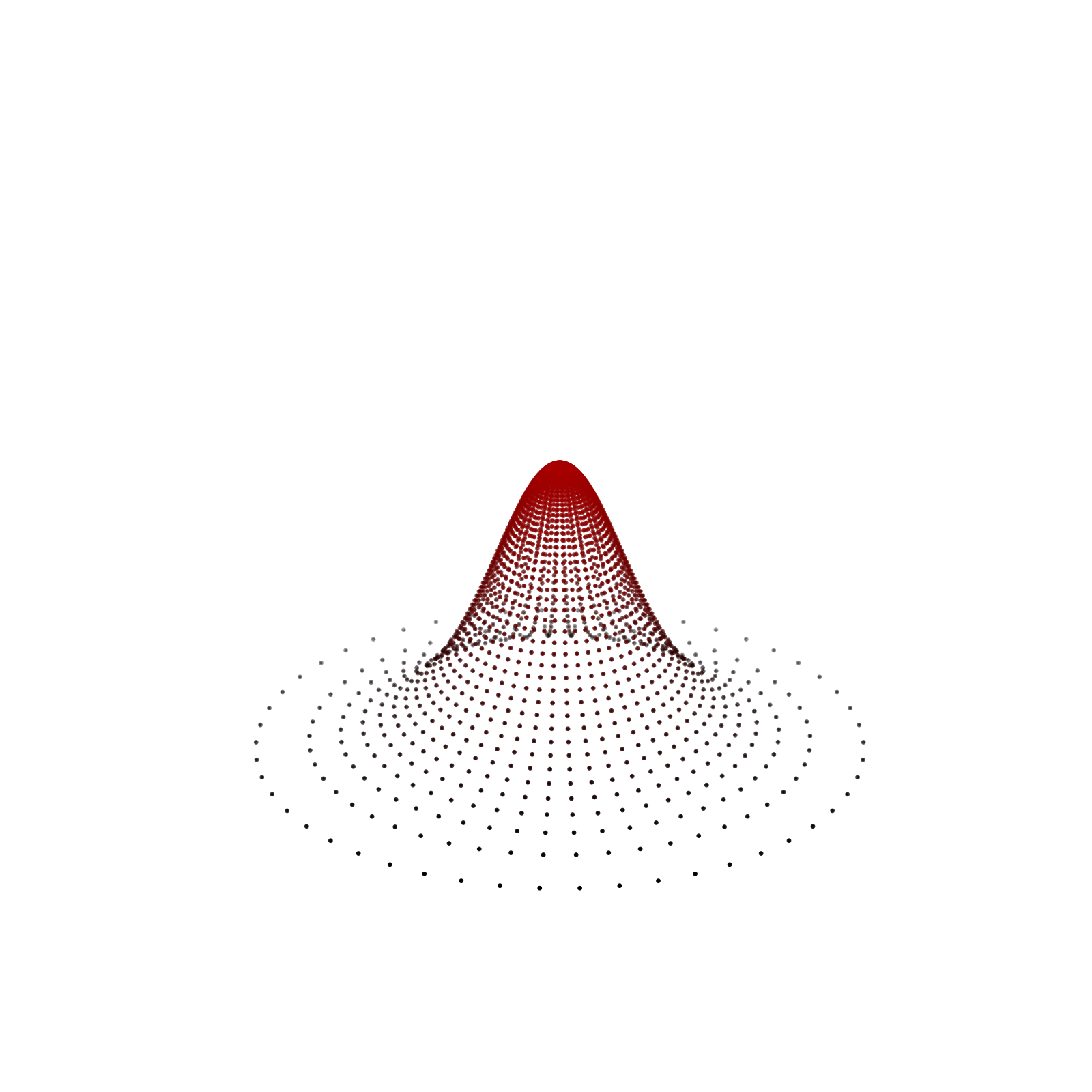}}
  \\[2pt]
  \subcaptionbox{Non-TC\label{fig:H2_std}}
   [0.32\linewidth]{%
    \includegraphics[width=\linewidth,trim=120 110 120 119,clip]{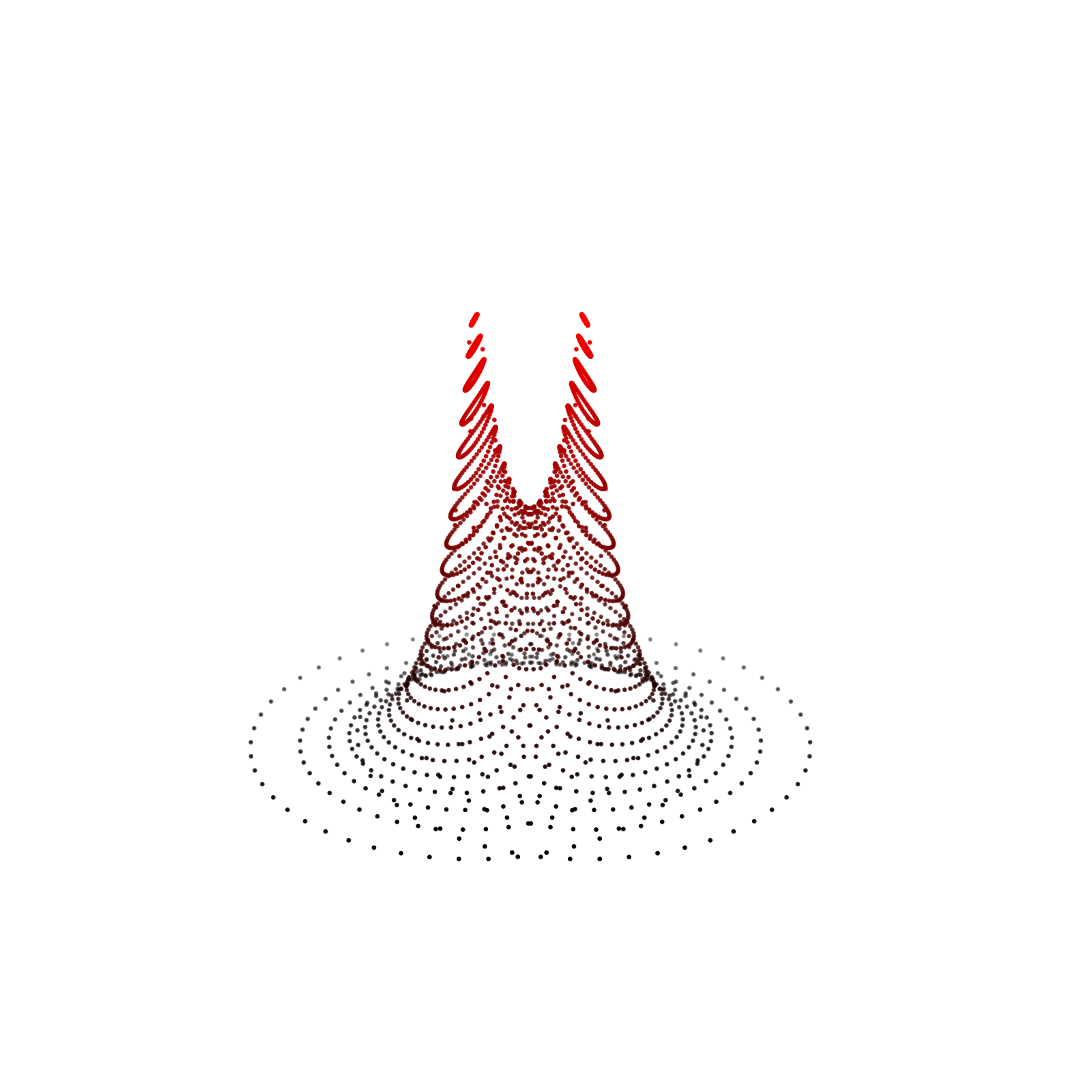}}
  &
  \subcaptionbox{$\mu_{ne}=3$\label{fig:H2_mu3}}
   [0.32\linewidth]{%
    \includegraphics[width=\linewidth,trim=120 110 120 119,clip]{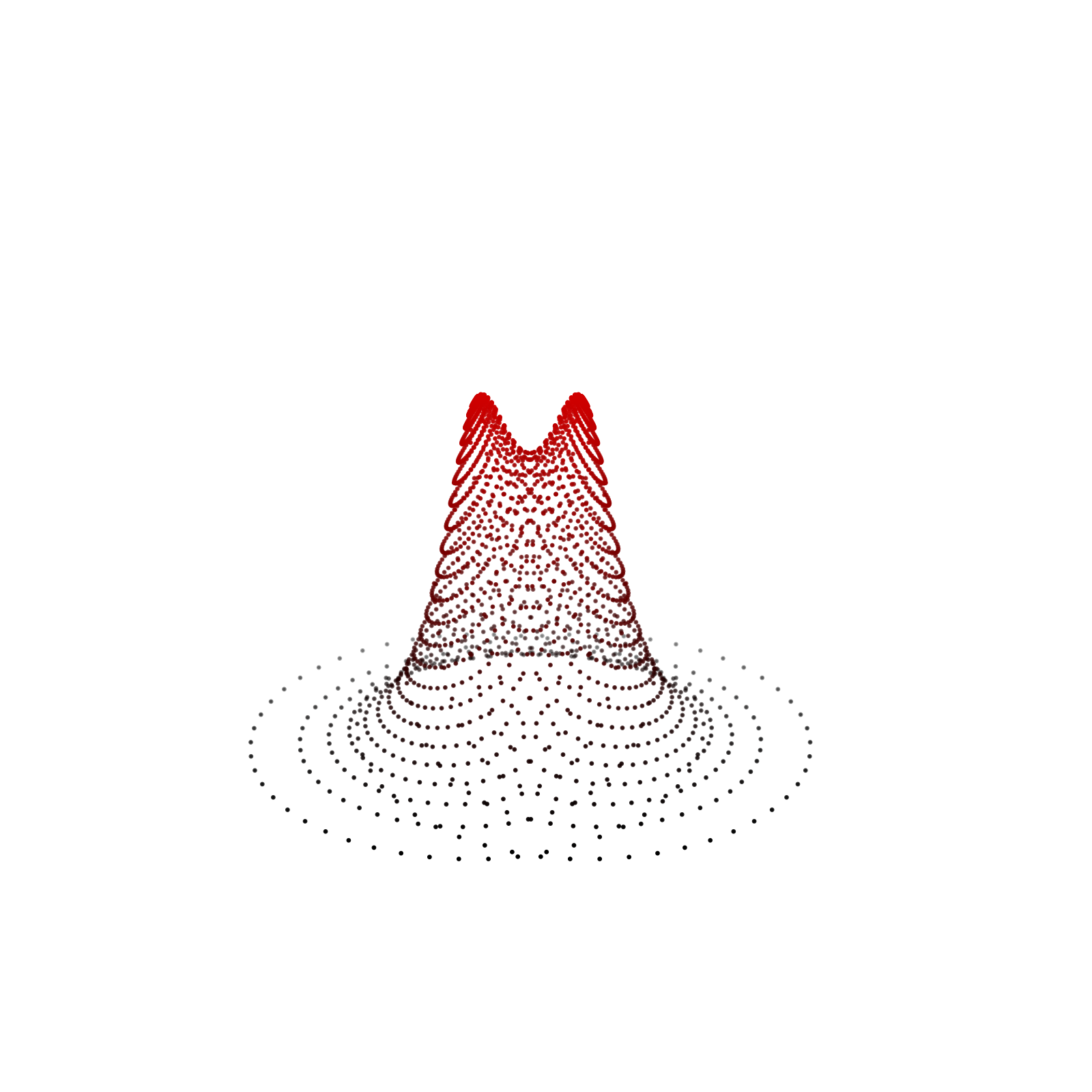}}
  &
  \subcaptionbox{$\mu_{ne}=1$\label{fig:H2_mu1}}
   [0.32\linewidth]{%
    \includegraphics[width=\linewidth,trim=120 110 120 119,clip]{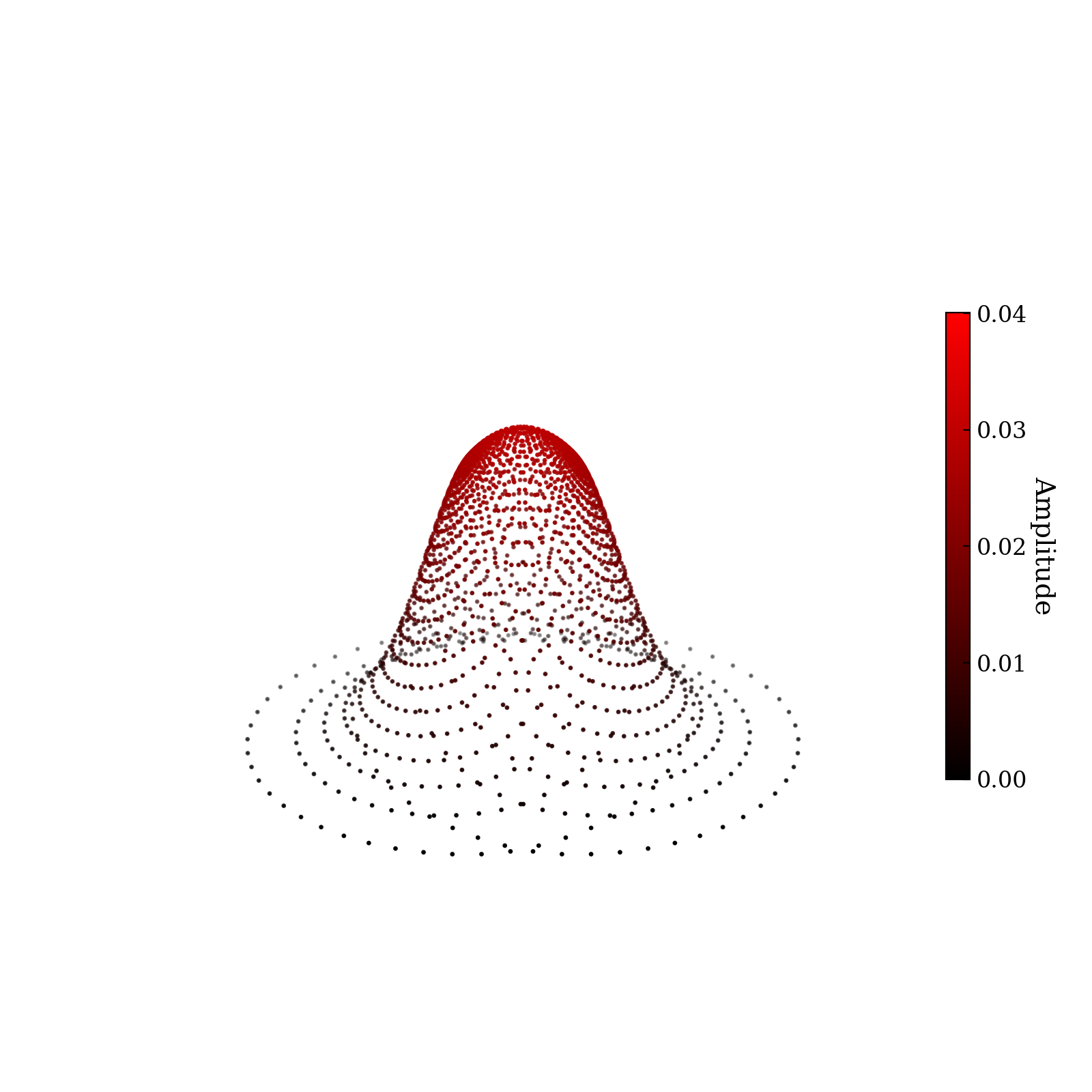}}
 \end{tabular}

 \caption{Sketches of ground‐state wavefunctions on a 2-D grid for single-center (top) and two-center (bottom) systems. Columns show the non-transcorrelated case (left) and transcorrelated cases for different values of $\mu_{ne}$.}
 \label{fig:combined_scatter}
\end{figure}

\begin{figure}[h]
 \centering
 \includegraphics[width=\linewidth]{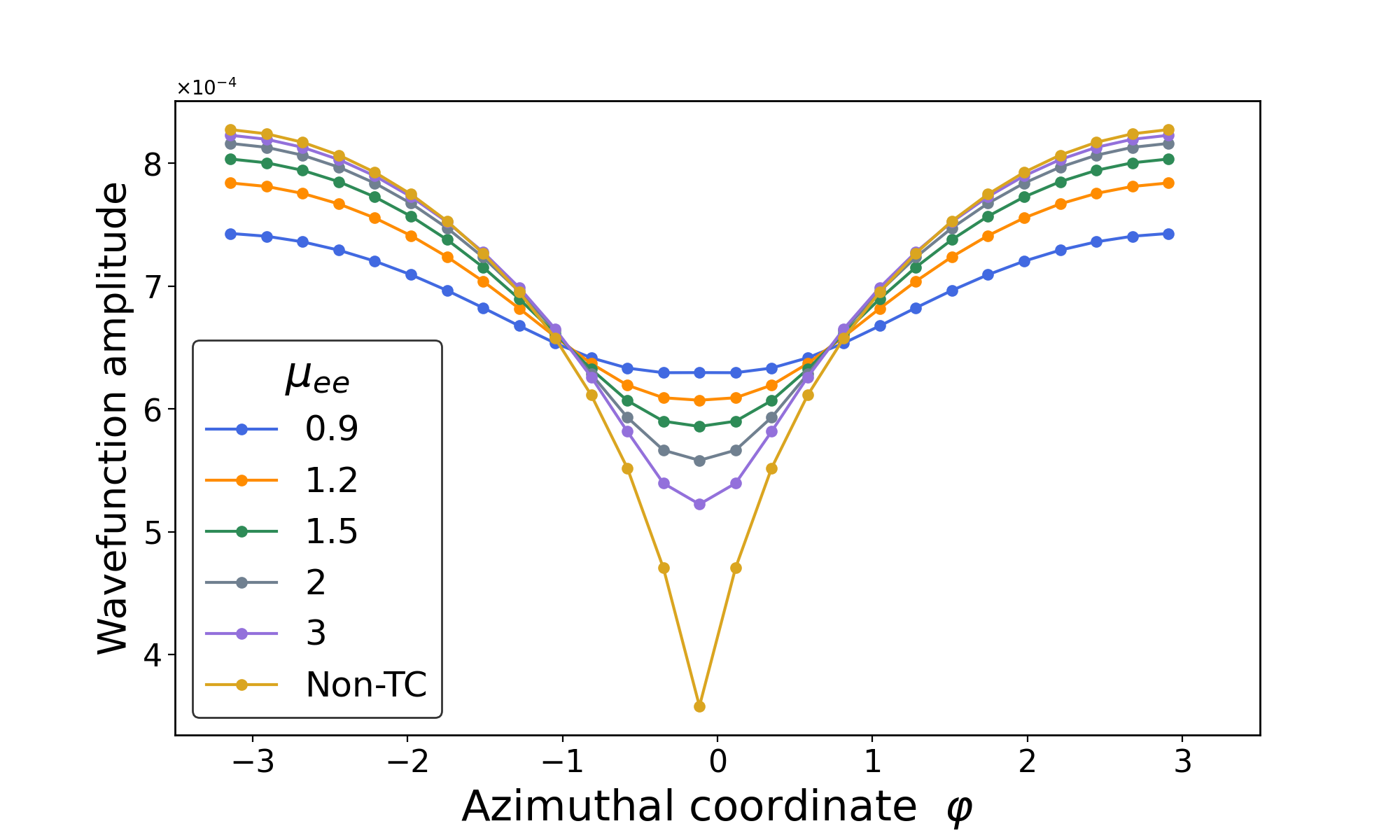}
 \caption{Cut of the ground-state helium wavefunction sampled along the angular coordinate \(\varphi\) on a circle of radius \(0.5\) bohr centered on the nucleus, with the second electron fixed near \(\varphi=0\). The amplitude is plotted on a Gauss--Legendre angular grid using the Voronoi finite-volume discretization. \CF{Decreasing} the transcorrelation parameter \(\mu_{ee}\) progressively smooths the electron--electron cusp.}
 \label{fig:eecusps}
\end{figure}

\noindent
\CF{Turning now to a two-electron system, Fig.~\ref{fig:eecusps} shows a cut of the ground-state wavefunction of the helium atom using the standard and transcorrelated Hamiltonian in the Voronoi basis as introduced in Sec.~\ref{sec:space_discr} and ~\ref{sec:TCHam}. The wavefunction is sampled along the angular coordinate on the radial shell \(r=0.5\) bohr, with the second electron fixed near \(\varphi=0\).} For the sake of clarity, the amplitudes were obtained with a Gauss–Legendre angular grid, whose separable $\theta,\varphi$ structure makes it straightforward to scan a 1-D cut along $\varphi$, albeit less optimal than Lebedev angular grid (as highlighted in the appendix \ref{SI-fursim}). In the non-transcorrelated case, the curve shows the familiar sharp cusp anticipated by the electron–electron coalescence condition. As the Jastrow parameter $\mu_{ee}$ decreases, the TC transformation progressively smooths this cusp, exactly as predicted by theory and reproducing the results of~\cite{Giner_2021}. \CF{Choosing a Jastrow parameter that is too small can make electron coalescence artificially favorable, resulting in an unphysical maximum of the transformed wavefunction at \(r_{12}=0\)~\cite{Giner_2021}.} Although the transcorrelated Hamiltonian remains isospectral, it is usually preferred to use parameters that preserve a physically meaningful wavefunction. 

\begin{figure}[h]
  \centering
  \includegraphics[width=\linewidth]{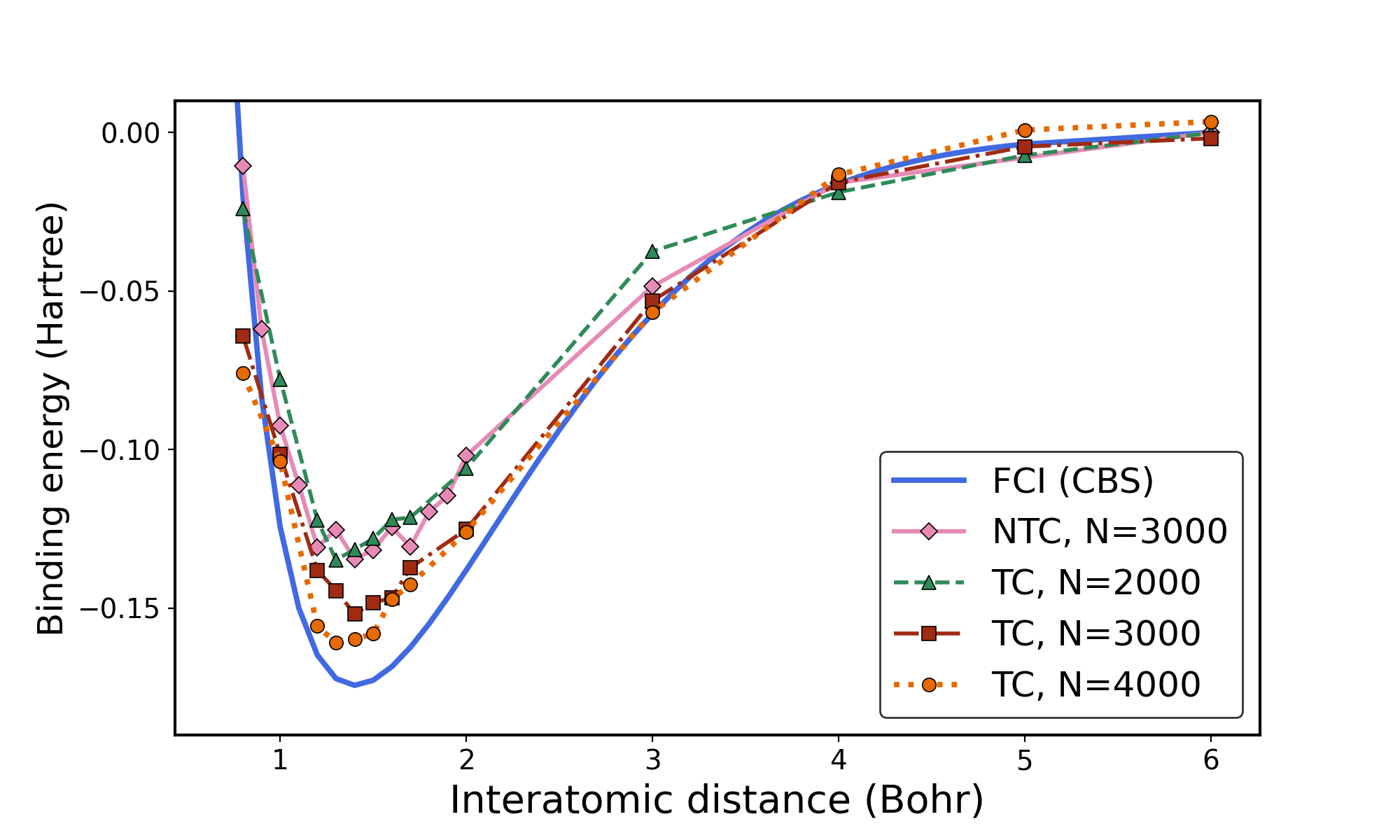}
\caption{\CF{Dissociation curve of H$_2$: binding energy versus internuclear distance for the transcorrelated (TC) method on nonuniform grids, compared to a non-transcorrelated (NTC) calculation and to FCI at the CBS limit. The grids contain 2000, 3000, and 4000 total points, obtained from 20, 30, and 40 radial points with 50 angular quadrature points per atom. The NTC reference is shown for the 3000-point grid. Classically, these correspond to Hamiltonian matrices of dimensions \(4\), \(9\), and \(16\) million, respectively, while a first-quantized quantum encoding would require only \(\lceil\log_2 N\rceil\) qubits per electron, i.e. \(11\)--\(12\) qubits per electron for these grids.}}
  \label{fig:diss}
\end{figure}
\CF{In Fig.~\ref{fig:diss}, we report the H$_2$ potential-energy curve as a function of the internuclear distance \(R\), using the multicenter adaptive grids introduced in Sec.~\ref{sec:space_discr}. The plotted energy is given relative to the dissociation limit,\(E_{\rm rel}(R)=E(\mathrm{H}_2;R)-2E(\mathrm{H}),\) where \(E(\mathrm{H})\) is the energy of an isolated hydrogen atom.} A grid of roughly 2 000 points, while capturing the qualitative shape, remains far from the exact binding energy; increasing the grid density systematically incorporates more electronic correlation and brings the curve closer to the FCI–CBS reference. \CF{The non-transcorrelated calculation on the 3000-point grid gives an accuracy comparable to the transcorrelated calculation on the 2000-point grid, illustrating the basis-set improvement brought by the transcorrelated transformation.}
An important aspect of the present simulations is that the energies are size‐consistent, as evidenced by the dissociation limit correctly approaching zero as the distance $R$ increases. While increasing the radial resolution systematically gives better results at a given geometry, the dissociation curves lack some perfectly smooth behavior across the stretching of the bond. We believe the observed kinks occur because a slight change in bond length can completely reshuffle the Voronoi cells associated to our non-overlapping multicenter grid. This abrupt reorganization sometimes creates locally favorable resolution and sometimes less favorable, directly affecting the computed energy. Further basis‐convergence simulations are presented in Appendix \ref{SI-fursim}. 

\section{Discussion and outlook}

Real‐space methods that explicitly track $\eta$ electrons on an $N$-point grid remain rare in standard computational chemistry because they naively require an $N^\eta$-dimensional Hilbert space. Using quantum devices, however, this Hilbert space can be encoded using $\eta \log N$ qubits, making fine spatial discretizations (large $N$) regimes attractive. For example, our H$_2$ simulation involved diagonalizing matrices as large as 16 million $\times$ 16 million, while the equivalent quantum encoding would use only 12 qubits per electron (4 096 grid points), and every additional pair of qubits doubles the spatial resolution. Since high-accuracy DFT grids rarely exceed $10^5$ points per atom, a ten-atom calculation would only require on the order of 20 qubits per electron. This provides an intuitive picture of how compactly the wavefunction can be stored while approaching the continuum limit. Prior first-quantized QPE analyses often rely on uniform grids or plane waves where the Coulomb singularity is often replaced by a bounded soft-Coulomb potential to make the cost analysis tractable. Provided the softening parameters are chosen as a valid approximation to the physical Coulomb problem, resolving the wavefunction near coalescence still induces large spatial-resolution overheads.
In this paper, we integrate nonuniform, molecular adaptive grids into a first‐quantized, real‐space quantum computing framework for electronic structure calculations. Further in this direction, we introduce the use of a transcorrelated Hamiltonian, that eliminates Coulombic singularities and corresponding cusps in the eigenvectors. Numerical validation on He and H$_2$ offers a realistic benchmark on atomic and molecular systems, featuring electron–nuclear and electron–electron cusps, as well as both dynamic and static correlation regimes, and the successful H$_2$ dissociation suggests that the present real-space, transcorrelated framework could handle other more complex chemical systems. 
Although detailed scaling analyses and quantitative comparisons to other methods are left for future work, this study establishes a flexible, first‐step approach that can be refined and extended in multiple directions. 

In the present work, we used a rather minimal set-up both in terms of the definition of the Voronoi cells and the correlation factor, which leaves a substantial room for improvement that we now outline. Classical computational chemistry has long optimized DFT integration grids and transcorrelated Hamiltonian forms, with Jastrow parameter selection being its own discipline. Likewise, optimizing grids and Jastrow choices on our numerical scheme should improve these baseline simulations. The Voronoi finite‐volume discretization allows to express differential operators in the basis on \textit{any} input grid, producing a Hamiltonian that is sparse yet has generally unstructured off‐diagonal elements. Any such operators can be block encoded via Pauli-LCU expansion, but the brute force classical preprocessing to find the coefficients scales with the nonzero elements in the Hamiltonian, thus exponentially in system size, making the data loading step a serious bottleneck. \CF{Since the molecular grids introduced here are organized as multicenter shells, they retain a degree of geometric regularity that could be exploited algorithmically. Beyond arithmetic and Pauli-based data loading, a natural direction is to encode real-space quantities through structured functional representations, such as polynomial, Walsh, or Fourier expansions, which are increasingly amenable to efficient quantum loading and synthesis~\cite{sunderhauf2024block,zylberman2025efficient,Berry_2019,claudon2026logarithmic,zylberman2024efficient,motlagh2024generalized}. Such developments would be particularly timely in light of the growing interest in grid-based quantum simulation of Coulombic systems~\cite{feng2026improved,eklund2026end,pocrnic2026efficient}.}

\noindent
\CF{This work contributes to the growing effort to exploit transcorrelated Hamiltonians in quantum algorithms~\cite{mcardle2020improving,dobrautz2023ab,dobrautz2024toward,sokolov2022orders,uvarov2025accuracy}. Since transcorrelation is intrinsically formulated in real space, adaptive real-space encodings combined with non-Hermitian eigenvalue processing provide a particularly direct route for future fault-tolerant implementations.}
\\[0.4cm]
\section*{Data Availability}
The Davidson diagonalisation routine used in this work is available in the open-source \emph{Quantum Package 2} (QP2) software suite at \url{https://quantumpackage.github.io/qp2/}. Voronoi diagram computations were performed using the publicly available \emph{Qhull} library (\url{http://www.qhull.org/}), and the full generation of the discretized Hamiltonians can be found at \url{https://doi.org/10.5281/zenodo.21640304}. No additional custom software was developed for this study.

\section*{Acknowledgments}
The authors wish to thank Yvon Maday, Igor Chollet, Pierre Monmarché and Solal Perrin-Roussel for fruitful discussions.
\section*{Funding}
This work has received funding from the European Research Council (ERC) under the European Union's Horizon 2020 research and innovation program (grant agreement No 810367), project EMC2 (JPP). Support from the PEPR EPIQ - Quantum Software (ANR-22-PETQ-0007, JPP) and HQI (JPP) programs is acknowledged.

\clearpage
\appendix
\onecolumngrid

\section{Block-Encoding of the Operators}
\subsection{Standard Molecular Hamiltonian}
The block-encoding of the molecular Hamiltonian can be constructed using the systematic Pauli LCU decomposition in the first quantization, as introduced in~\cite{Georges_2025,georges2025quantum}, and is valid for arbitrary basis sets, including one that diagonalizes the potential energy. We quote the expression in~\cite{georges2025quantum} for the Pauli LCU decomposition of a Hamiltonian that can be decomposed into a one-body term and a diagonal two-body term:
\begin{equation}
\mathbf{\bar{H}}=\sum_{m,n=0}^{N-1}\omega_{mn}\sum_{i=0}^{\eta-1}\prod_{b=0}^{\log N-1}\mathbf{X}_{i\log N+b}^{m_b}\mathbf{Z}_{i\log N+b}^{n_b} + \frac{1}{2}\sum_{m,p=0}^{N-1}\gamma_{mp}\sum_{i\neq j}^{\eta -1}\prod_{b=0}^{\log N -1} \mathbf{Z}_{i\log N +b}^{m_b} \mathbf{Z}_{j\log N +b}^{p_b}
\end{equation}
where $m_b$ is the $b$-th bit in the binary representation of $m$, likewise for $n_b$ and $p_{b}$, and $\mathbf{X}_i$ and $\mathbf{Z}_i$ are Pauli operators acting on the $i$-th qubit:
\begin{equation}
    \mathbf{X}_i = \mathbf{I}_{2^i}\otimes\mathbf{X}\otimes \mathbf{I}_{2^{\eta \log N -i-1}}, \qquad \mathbf{Z}_i = \mathbf{I}_{2^i}\otimes\mathbf{Z}\otimes \mathbf{I}_{2^{\eta \log N -i-1}}.
\end{equation}
The explicit expressions of the coefficients are
\begin{equation}
  \omega_{mn}=\frac{1}{N}\sum_{x=0}^{N-1} (-1)^{x \odot n}\mathbf{\bar{T}}_{m\oplus x,x}, \qquad
  \gamma_{mp}=\frac{1}{N^2}\sum_{x,y=0}^{N-1}(-1)^{(m \odot x)+(y \odot p)}\mathbf{W}_{xy}
  \label{SI-omega_mp}
\end{equation}
In Eqs. \eqref{SI-omega_mp}, $\oplus$ is the bitwise XOR operation and $\odot$ is the bitwise dot product. The expressions of $\mathbf{\bar{T}}_{mn}$ and $\mathbf{W}_{mp}$ are given by equations \eqref{symmoneb}, \eqref{oneb} and \eqref{twob}. One can reduce the number of coefficients to load by neglecting terms with identity Pauli strings, since they only contribute to a constant shift of the eigenvalue. This is done by not considering coefficients for which both indices are zero. Also, one can incorporate repeating strings in the one-body and two-body terms into a single coefficient, which corresponds to that of the string containing a $\mathbf{Z}$ gate on only one of the qubits. Noticing that $\gamma_{mp}=\gamma_{pm}$, we can rewrite the LCU decomposition of the Hamiltonian with these modifications:
\begin{equation}
\mathbf{\bar{H}}_{LCU}=\sum_{m,n=0}^{N-1}\omega_{mn}'\sum_{i=0}^{\eta-1}\prod_{b=0}^{\log N-1}\mathbf{X}_{i\log N+b}^{m_b}\mathbf{Z}_{i\log N+b}^{n_b} + \frac{1}{2}\sum_{m,p=1}^{N-1}\gamma_{mp}'\sum_{i\neq j}^{\eta -1}\prod_{b=0}^{\log N -1} \mathbf{Z}_{i\log N +b}^{m_b} \mathbf{Z}_{j\log N +b}^{p_b}
\end{equation}
where 
\begin{equation}
\omega_{mn}'=
  \begin{cases}
\displaystyle 0& \text{if } m=n=0, \\
\omega_{0n} + (N-1)\gamma_{0n} & \text{if } m=0 \text{, } n\neq 0, \\
\omega_{mn} &\text{otherwise},
\end{cases}
\qquad \qquad
\gamma_{mp}'=
  \begin{cases}
\displaystyle 0& \text{if } m=0 \text{ or } p=0,\\
\gamma_{mp} &\text{otherwise}.
\end{cases}
\end{equation}
The subscript LCU is added to differentiate between the true Hamiltonian and the one without identity strings. 
\subsection{Transcorrelated Molecular Hamiltonian}
We adopt a similar approach to block encode the transcorrelated Hamiltonian using an adapted LCU decomposition:
\begin{equation}
\small
\begin{split}
  &\mathbf{\tilde{H}} = \sum_{m,n=0}^{N-1}\tilde{\omega}_{mn}\sum_{i=0}^{\eta-1}\prod_{b=0}^{\log N -1}\mathbf{X}_{i\log N +b}^{m_b}\mathbf{Z}_{i\log N +b}^{n_b}
  +\frac{1}{2}\sum_{m,n,p,q=0}^{N-1}\tilde{\gamma}_{mnpq}\sum_{i\neq j}^{\eta-1}\prod_{b=0}^{\log N -1}\mathbf{X}_{i\log N +b}^{m_b}\mathbf{Z}_{i\log N +b}^{n_b}
  \mathbf{X}_{j\log N +b}^{p_b}\mathbf{Z}_{j\log N +b}^{q_b} \\
  &- \frac{1}{3}\sum_{m,p,t=0}^{N-1}\tilde{\zeta}_{mpt}\sum_{i\neq j \neq k}^{\eta-1}\prod_{b=0}^{\log N -1} \mathbf{Z}_{i\log N +b}^{m_b} \mathbf{Z}_{j\log N +b}^{p_b} \mathbf{Z}_{k\log N +b}^{t_b} 
\end{split}
\end{equation}
The explicit expressions for the coefficients are given by:
\begin{equation}
  \tilde{\omega}_{mn}=\frac{1}{N}\sum_{x=0}^{N-1}(-1)^{x\odot n} \mathbf{\tilde{T}}_{m\oplus x,x}
\end{equation}
\begin{equation}
  \tilde{\gamma}_{mnpq} = \frac{1}{N^2}\sum_{x,y=0}^{N-1}(-1)^{(x\odot n)+(y\odot q)}\mathbf{\tilde{W}}_{m\oplus x,x,p \oplus y,y}
\end{equation}
\begin{equation}
  \tilde{\zeta}_{mpt}=\frac{1}{N^3}\sum_{x,y,z=0}^{N-1} (-1)^{(x\odot m) + (y\odot p) + (z\odot t)}\mathbf{\tilde{B}}_{mpt}
\end{equation}

\begin{equation}
  \begin{split}
  \mathbf{\tilde{H}}_{LCU} &= \sum_{m,n=0}^{N-1}\tilde{\omega}_{mn}'\sum_{i=0}^{\eta-1}\prod_{b=0}^{\log N -1}\mathbf{X}_{i\log N +b}^{m_b}\mathbf{Z}_{i\log N +b}^{n_b}
  +\frac{1}{2}\sum_{m,n,p,q=0}^{N-1}\tilde{\gamma}_{mnpq}'\sum_{i\neq j}^{\eta-1}\prod_{b=0}^{\log N -1}\mathbf{X}_{i\log N +b}^{m_b}\mathbf{Z}_{i\log N +b}^{n_b}
  \mathbf{X}_{j\log N +b}^{p_b}\mathbf{Z}_{j\log N +b}^{q_b} \\
  &- \frac{1}{3}\sum_{m,p,t=1}^{N-1}\tilde{\zeta}_{mpt}'\sum_{i\neq j \neq k}^{\eta-1}\prod_{b=0}^{\log N -1} \mathbf{Z}_{i\log N +b}^{m_b} \mathbf{Z}_{j\log N +b}^{p_b} \mathbf{Z}_{k\log N +b}^{t_b} 
\end{split} \label{SI-transLCU}
\end{equation}
where 
\begin{equation}
  \tilde{\omega}_{mn}'=
  \begin{cases}
\displaystyle 0 &\text{if } m=n=0\\
\tilde{\omega}_{0n}+\frac{1}{2}(N-1)(\tilde{\gamma}_{0n00}+\tilde{\gamma}_{000n})-\frac{1}{3}(N-1)^{2}(\tilde{\zeta}_{n00}+2\tilde{\zeta}_{00n})& \text{if }m=0\\
\tilde{\omega}_{mn}+\frac{1}{2}(N-1)(\tilde{\gamma}_{mn00}+\tilde{\gamma}_{00mn})& \text{otherwise}
\end{cases} 
\end{equation}

\begin{equation}
  \tilde{\gamma}_{mnpq}'=
  \begin{cases}
\displaystyle 0 &\text{if } m=n=0 \text{ or } p=q=0\\
\tilde{\gamma}_{0n0q} - \frac{1}{3}(N-1)(\tilde{\zeta}_{0nq}+2\tilde{\zeta}_{nq0})& \text{if }m=p=0 \\
\tilde{\gamma}_{mnpq} &\text{otherwise}
\end{cases} 
\end{equation}

\begin{equation}
  \tilde{\zeta}_{mpt}'=
  \begin{cases}
\displaystyle 0 &\text{if } m=0 \text{ or } p=0 \text{ or } t=0\\
\tilde{\zeta}_{mpt}& \text{otherwise}
\end{cases}  \label{SI-lastLCU}
\end{equation}

\noindent
Similarly to the two-body LCU expansion derived in~\cite{georges2025quantum}, one can straightforwardly write one for the three-body term. Note that ${\mathbf{\tilde{T}}_{mn}}$ and ${\mathbf{\tilde{W}}_{mnpq}}$ contain non-Hermitian terms, and the latter two-body matrix is no longer diagonal. As before, we remove identity Pauli strings and repeating terms to get fewer coefficients to load. \CF{For the three-body term, the symmetry \(\tilde{\zeta}_{mpt}=\tilde{\zeta}_{mtp}\) can be exploited to obtain Eqs.~\eqref{SI-transLCU}--\eqref{SI-lastLCU}.}\\

\section{Expression of Differential Operators in the Voronoi Finite Volume Scheme} \label{SI-VFV}

Consider the molecular Schrödinger equation in \eqref{ps1} and \eqref{ps2}:
\begin{equation}
 \sum_{i=0}^{\eta-1}\left[ -\frac{1}{2}\nabla_{\mathbf{x}_i}^{2}\Psi(\ldots,\mathbf{x}_i,\ldots) -\sum_{\alpha=0}^{M-1}\frac{Z_\alpha}{|\mathbf{x}_i-\mathbf{R}_\alpha|} \Psi(\ldots,\mathbf{x}_i,\ldots) +\sum_{j>i}^{\eta-1}\frac{1}{|\mathbf{x}_i-\mathbf{x}_j|} \Psi(\ldots,\mathbf{x}_i,\ldots) \right] =E\Psi(\ldots,\mathbf{x}_i,\ldots).  \label{SI-ps3}
\end{equation}

and suppose the entire space of interest is filled with arbitrarily distributed points $\{\mathbf{r}_m\}_{m=0}^{N-1}$, forming a Voronoi diagram in $\mathbb{R}^{3}$. We take the volume integral over the Voronoi cells $\{\text{Vor}(\mathbf{r}_{m_i})\}_{i=0}^{\eta-1}$ occupied by each electron on both sides:
\begin{equation}
\begin{split}
  &\sum_{i=0}^{\eta-1} \Big{(} -\frac{1}{2} \int_{\text{Vor}(\mathbf{r}_{m_0})}\ldots \int_{\text{Vor}(\mathbf{r}_{m_i})}\ldots \int_{\text{Vor}(\mathbf{r}_{m_{\eta-1}})}\nabla^{2}_{\mathbf{x}_i} \Psi(\mathbf{x}_0,\ldots,\mathbf{x}_i,\ldots,\mathbf{x}_{\eta-1})\,d\mathbf{x}_0 \ldots d\mathbf{x}_{i}\ldots d\mathbf{x}_{\eta-1}\\
  & -\sum_{\alpha=0}^{M-1} \int_{\text{Vor}(\mathbf{r}_{m_0})}\ldots \int_{\text{Vor}(\mathbf{r}_{m_i})}\ldots \int_{\text{Vor}(\mathbf{r}_{m_{\eta-1}})}
  \frac{\mathcal{Z}_\alpha}{|\mathbf{x}_i - \mathbf{R}_\alpha|} \Psi(\mathbf{x}_0,\ldots,\mathbf{x}_i,\ldots,\mathbf{x}_{\eta-1})\,d\mathbf{x}_0 \ldots 
  d\mathbf{x}_{i}\ldots d\mathbf{x}_{\eta-1}\\
  &+ \sum_{\substack{j>i}}^{\eta-1} \int_{\text{Vor}(\mathbf{r}_{m_0})}\ldots \int_{\text{Vor}(\mathbf{r}_{m_i})}\ldots \int_{\text{Vor}(\mathbf{r}_{m_{\eta-1}})}
  \frac{1}{|\mathbf{x}_i - \mathbf{x}_j|}\Psi(\mathbf{x}_0,\ldots,\mathbf{x}_i,\ldots,\mathbf{x}_{\eta-1}) \,d\mathbf{x}_0 \ldots d\mathbf{x}_{i}\ldots d\mathbf{x}_{\eta-1} \Big{)}\\
  &= E\int_{\text{Vor}(\mathbf{r}_{m_0})}\ldots \int_{\text{Vor}(\mathbf{r}_{m_i})}\ldots \int_{\text{Vor}(\mathbf{r}_{m_{\eta-1}})}\Psi(\mathbf{x}_0,\ldots,\mathbf{x}_i,\ldots,\mathbf{x}_{\eta-1})\,d\mathbf{x}_0 \ldots d\mathbf{x}_{i}\ldots d\mathbf{x}_{\eta-1}
\end{split} \label{SI-scheme1}
\end{equation}
According to the divergence theorem applied to a gradient field, we can express the main integral of the first term inside the sum as
\begin{equation}
\begin{split}
  \int_{\text{Vor}(\mathbf{r}_{m_i})}\nabla^{2}_{\mathbf{x}_i} \Psi(\ldots,\mathbf{x}_i,\ldots)\,d\mathbf{x}_i = \int_{\partial\text{Vor}(\mathbf{r}_{m_i})} \nabla_{\mathbf{x}_i}\Psi(\ldots,\mathbf{x}_i,\ldots)\cdot d \mathbf{s}_{m_i}
  = \sum_{n_{i}\in\Lambda(m_{i})} \int_{\Gamma_{m_{i}n_{i}}} \nabla_{\mathbf{x}_i}\Psi(\ldots,\mathbf{x}_i,\ldots)\cdot \mathbf{\hat{r}}_{m_{i}n_{i}} \,d \mathbf{\sigma}_{m_{i}n_{i}} \label{SI-mainint}
\end{split}
\end{equation}
where $d \mathbf{s}_{m_i}$ is the unit normal area vector of the boundary $\partial \text{Vor}(\mathbf{r}_{m_i})$ of the Voronoi cell $\text{Vor}(\mathbf{r}_{m_i})$. In the last equality, we decomposed this boundary as the union of the Voronoi facets $\Gamma_{m_{i}n_{i}}$ with all surrounding neighbors to $\text{Vor}(\mathbf{r}_{m_i})$, \textit{ie} $\partial \text{Vor}(\mathbf{r}_{m_i})=\bigcup_{{n_i}\in\Lambda(m_{i})}\Gamma_{m_{i}n_{i}}$. Note that the unit normal vector $\mathbf{\hat{r}}_{m_{i}n_{i}} = (\mathbf{r}_{n_i} - \mathbf{r}_{m_i})/|\mathbf{r}_{n_i} - \mathbf{r}_{m_i}|$ is a normal vector of the facet $\Gamma_{m_{i}n_{i}}$. Now, we divide equation \eqref{SI-scheme1} by the product of the volumes of the Voronoi cells $\{v_{m_i}\}_{i=1}^{\eta}$, and take the limit when each one goes to zero. This consists in evaluating the integrands at the discrete points, multiplied by the integration measure and divided by the product of volumes. For one of the terms in the sum of equation \eqref{SI-mainint}, this gives
\begin{equation}
  \lim _{v_{m_i}\rightarrow0}\frac{1}{v_{m_i}}\int_{\Gamma_{m_{i}n_{i}}} \nabla_{\mathbf{x}_i}\Psi(\mathbf{x}_1,\ldots,\mathbf{x}_i,\ldots,\mathbf{x}_\eta)\cdot \mathbf{\hat{r}}_{m_{i}n_{i}} \,d {\sigma}_{m_{i}n_{i}}=\frac{1}{v_{m_i}}\nabla_{\mathbf{r}_{m_i}}\Psi(\mathbf{x}_1,\ldots,\mathbf{x}_{i-1},\mathbf{r}_{m_i},\mathbf{x}_{i+1}.\ldots,\mathbf{x}_\eta)\cdot \mathbf{\hat{r}}_{m_{i}n_{i}}{\sigma}_{m_{i}n_{i}} \label{SI-lim}
\end{equation}
Proceeding similarly with all other terms, one gets
\begin{equation}
\begin{split}
  &\sum_{i=0}^{\eta-1} \Big{(} -\frac{1}{2v_{m_i}}\sum_{n_{i}\in\Lambda(m_{i})} \nabla_{\mathbf{r}_{m_i}}\Psi(\mathbf{r}_{m_0},\ldots,\mathbf{r}_{m_i},\ldots,\mathbf{r}_{m_{\eta-1}})\cdot \mathbf{\hat{r}}_{m_{i}n_{i}}\,\mathbf{\sigma}_{m_{i}n_{i}}
  - \sum_{\alpha=0}^{M-1} 
  \frac{\mathcal{Z}_\alpha}{|\mathbf{r}_{m_i} - \mathbf{R}_\alpha|} \Psi(\mathbf{r}_{m_0},\ldots,\mathbf{r}_{m_i},\ldots,\mathbf{r}_{m_{\eta-1}})\\
  &+ \sum_{\substack{j>i}}^{\eta-1} 
  \frac{1}{|\mathbf{r}_{m_i} - \mathbf{r}_{m_j}|} \Psi(\mathbf{r}_{m_0},\ldots,\mathbf{r}_{m_i},\ldots,\mathbf{r}_{m_{\eta-1}}) \Big{)}
  = E\Psi(\mathbf{r}_{m_0},\ldots,\mathbf{r}_{m_i},\ldots,\mathbf{r}_{m_{\eta-1}}) \label{SI-discrete}
\end{split}
\end{equation}
The continuous problem in \eqref{SI-ps3} has thus been reduced to a discretized scheme on the Voronoi diagram. We also approximate the directional derivative with a finite difference, whose error can easily be identified using Taylor's theorem:
\begin{equation}
  \nabla_{\mathbf{r}_{m_i}}\Psi(\mathbf{r}_{m_0},\ldots,\mathbf{r}_{m_i},\ldots,\mathbf{r}_{m_{\eta-1}})\cdot \mathbf{\hat{r}}_{m_{i}n_{i}} = \frac{\Psi(\mathbf{r}_{m_0},\ldots,\mathbf{r}_{n_i},\ldots,\mathbf{r}_{m_{\eta-1}})-\Psi(\mathbf{r}_{m_0},\ldots,\mathbf{r}_{m_i},\ldots,\mathbf{r}_{m_{\eta-1}})}{|\mathbf{r}_{n_i} - \mathbf{r}_{m_i}|} + \mathcal{O}(|\mathbf{r}_{n_i} - \mathbf{r}_{m_i}|)  \label{SI-directional}
\end{equation}
Combining equations \eqref{SI-discrete} and \eqref{SI-directional}, we find the expression of the discretized Laplacian operator in \eqref{laplacian} and that of the scalar potential energy operators.
For the gradient operator, similarly to what was done in \eqref{SI-mainint} and \eqref{SI-lim}, and using an alternative form of the divergence theorem:
\begin{equation}
\begin{split}
   &\lim _{v_{m_i}\rightarrow0}\frac{1}{v_{m_i}} \int_{\text{Vor}(\mathbf{r}_{m_i})} \nabla_{\mathbf{x}_i}\Psi(\ldots,\mathbf{x}_i,\ldots) \,d\mathbf{x}_i = \lim _{v_{m_i}\rightarrow0}\frac{1}{v_{m_i}} \int_{\partial\text{Vor}(\mathbf{r}_{m_i})}\Psi(\ldots,\mathbf{x}_i,\ldots)\big|_{\partial\text{Vor}(\mathbf{r}_{m_i})}\,d\mathbf{s}_{m_i} \\
   &=\sum_{n_{i}\in\Lambda(m_{i})} \lim _{v_{m_i}\rightarrow0}\frac{1}{v_{m_i}} \int_{\Gamma_{m_{i}n_{i}}} \Psi(\ldots,\mathbf{x}_i,\ldots)\big|_{\Gamma_{m_{i}n_{i}}} \mathbf{\hat{r}}_{m_{i}n_{i}} \,d\mathbf{\sigma}_{m_{i}n_{i}}\\
   &=\sum_{n_{i}\in\Lambda(m_{i})}\frac{1}{v_{m_i}} \Psi(\mathbf{r}_{m_1},\ldots,\mathbf{r}_{m_i},\ldots,\mathbf{r}_{m_\eta})\big|_{\Gamma_{m_{i}n_{i}}}\mathbf{\hat{r}}_{m_{i}n_{i}}\mathbf{\sigma}_{m_{i}n_{i}}
\end{split}
\end{equation}
We approximate the value of the wavefunction on the facet $\Gamma_{m_{i}n_{i}}$ by the average value
\begin{equation}
\Psi(\mathbf{r}_{m_0},\ldots,\mathbf{r}_{m_i},\ldots,\mathbf{r}_{m_{\eta-1}})\big|_{\Gamma_{m_{i}n_{i}}}\approx\frac{\Psi(\mathbf{r}_{m_0},\ldots,\mathbf{r}_{m_i},\ldots,\mathbf{r}_{m_{\eta-1}})+\Psi(\mathbf{r}_{m_0},\ldots,\mathbf{r}_{n_i},\ldots,\mathbf{r}_{m_{\eta-1}})}{2}
\end{equation}
Knowing that $\sum_{n_i}\mathbf{\hat{r}}_{m_{i}n_{i}}\sigma_{m_{i}n_{i}}=0$ for any convex polygon, the matrix form of the gradient operator projected on a unit vector $\hat{\mathbf{z}}$, \textit{ie} the matrix representation of the directional derivative along $\hat{\mathbf{z}}$, is given by:
\begin{equation}
\mathbf{D}^{(\hat{\mathbf{z}})}_{mn} =
\begin{cases}
\displaystyle \frac{\sigma_{mn}}{2 v_m} \frac{\mathbf{{r}}_{m}-\mathbf{r}_{n}}{|\mathbf{{r}}_{m}-\mathbf{r}_{n}|}\cdot \hat{\mathbf{z}} & \text{if } n \in \Lambda(m) \\[1ex]
0 & \text{otherwise} 
\end{cases} \label{SI-grad}
\end{equation}

\section{Initial state}
In second‐quantized encoding, qubits directly represent orbital occupations, making the Hartree–Fock (HF) determinant trivial to prepare. When stronger ground-state support is needed, various techniques exist to load correlated CI states within the second‐quantized formalism~\cite{feniou2024sparse, tubman2018postponing, fomichev2023initial}.
First‐quantized, real‐space encoding requires classically computed molecular orbitals to be evaluated over the chosen discretisation, and the corresponding amplitudes loaded to the computational basis of the qubit register. Then, a fermionic antisymmetrization circuit of size \(\mathcal O(\eta\log\eta\log N)\) and depth \(\mathcal O(\log\eta\log\log N)\) is applied to the single-electron registers. The key task is efficient loading of Gaussian orbital amplitudes onto the computational basis, a task that can be addressed through careful grid points / Voronoi cell ordering and advanced state loading routines~\cite{motlagh2024generalized,zylberman2024efficient,holmes2020efficientquantumcircuitsaccurate}. Likewise, CI wavefunctions can be prepared by first loading each Slater determinant with the above procedure and then assembling the CI state using the standard LCU approach with \texttt{PREPARE} and \texttt{SELECT} oracles. 
\label{SI-initial_state}

\section{Further Basis Convergence Simulations} \label{SI-fursim}

Figure~\ref{SI-fig:h_conv} compares the convergence of the hydrogen ground-state energy for three Becke–Lebedev radial exponents \(\nu=1,2,3\). Panel (a) shows the non-transcorrelated Hamiltonian, while (b) and (c) use TC Hamiltonians with \(\mu_{ne}=3\) and \(\mu_{ne}=1\), respectively. In each plot the error (in Hartree) is plotted on a logarithmic scale versus the total number of grid points \(N\). The blue band indicates chemical accuracy (\(\le 1~\text{mHa}\)) to the FCI-CBS. Introducing the Jastrow factor accelerates convergence, and reducing \(\mu_{ne}\) (\(\mu_{ne}=1\)) achieves the chemical-accuracy threshold several thousand points sooner than the non-TC case. Across all settings the grid with \(\nu=1\) consistently outperforms higher exponents. \CF{This suggests that, for this simple system, the radial grid need not be excessively concentrated near the nucleus. This may seem counterintuitive, since larger Becke exponents are commonly used in DFT quadrature. However, the present grid is not an integration rule: it defines the basis on which the differential operator is represented. There is therefore no reason for the exponent that is optimal for DFT quadrature to also be optimal for solving the Hamiltonian ground state problem.}

\begin{figure*}[h]
 \centering
 \subcaptionbox{Non-TC\label{SI-fig:h_conv_ntc}}
  [0.31\linewidth]{\includegraphics[width=\linewidth]{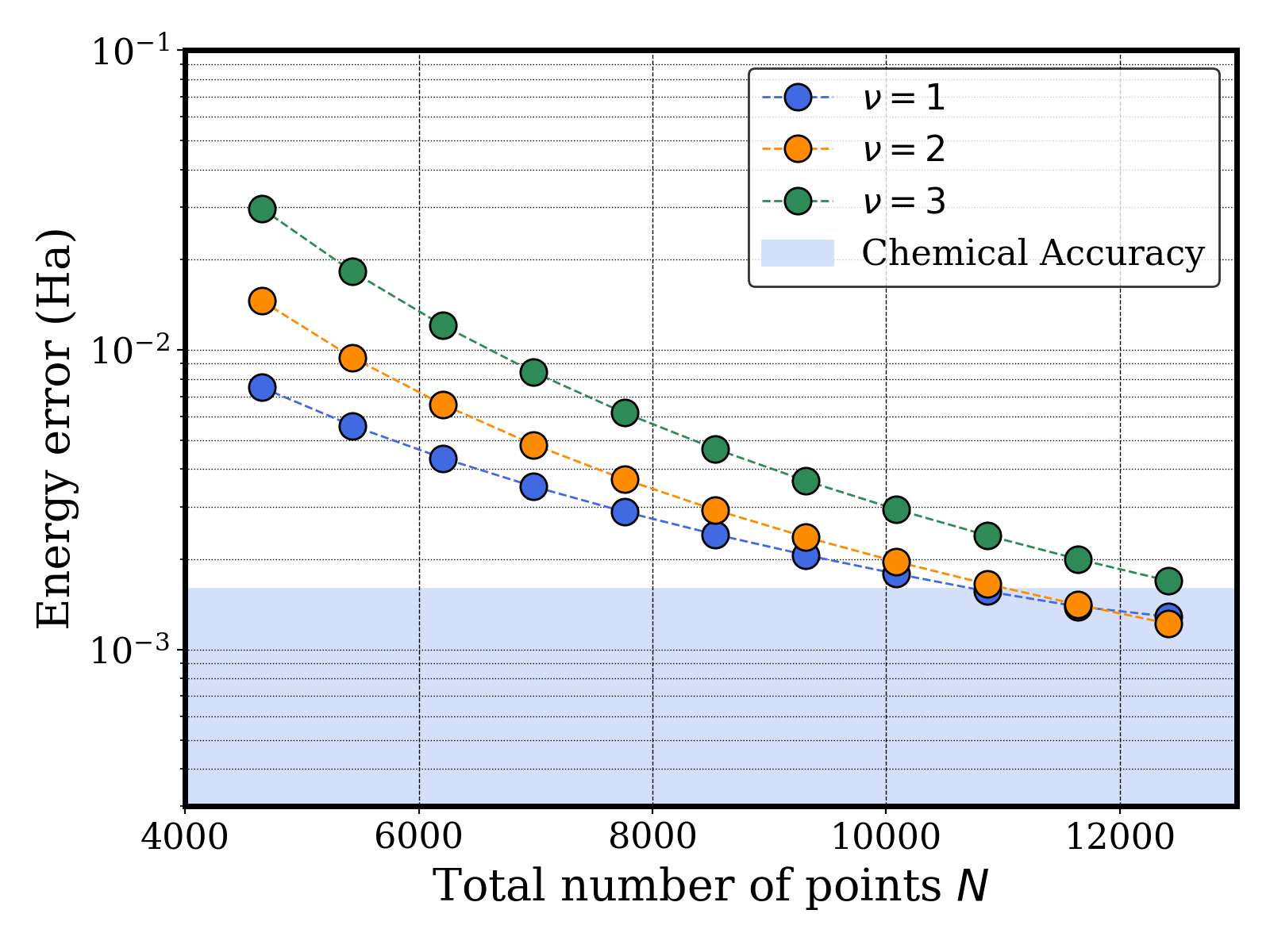}}
 \subcaptionbox{$\mu_{ne}=3$\label{SI-fig:h_conv_mu3}}
  [0.31\linewidth]{\includegraphics[width=\linewidth]{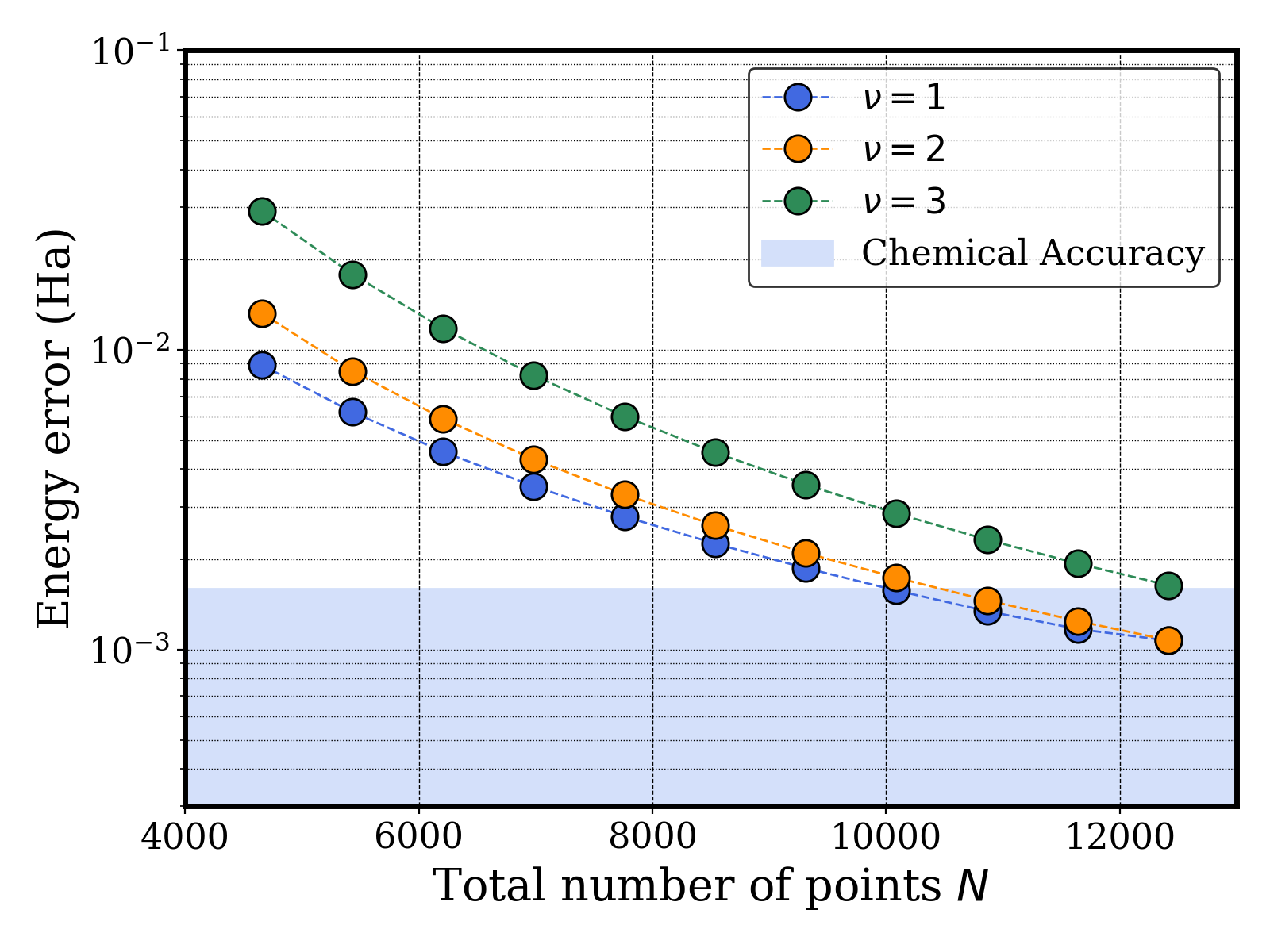}}
 \subcaptionbox{$\mu_{ne}=1$\label{SI-fig:h_conv_mu1}}
  [0.31\linewidth]{\includegraphics[width=\linewidth]{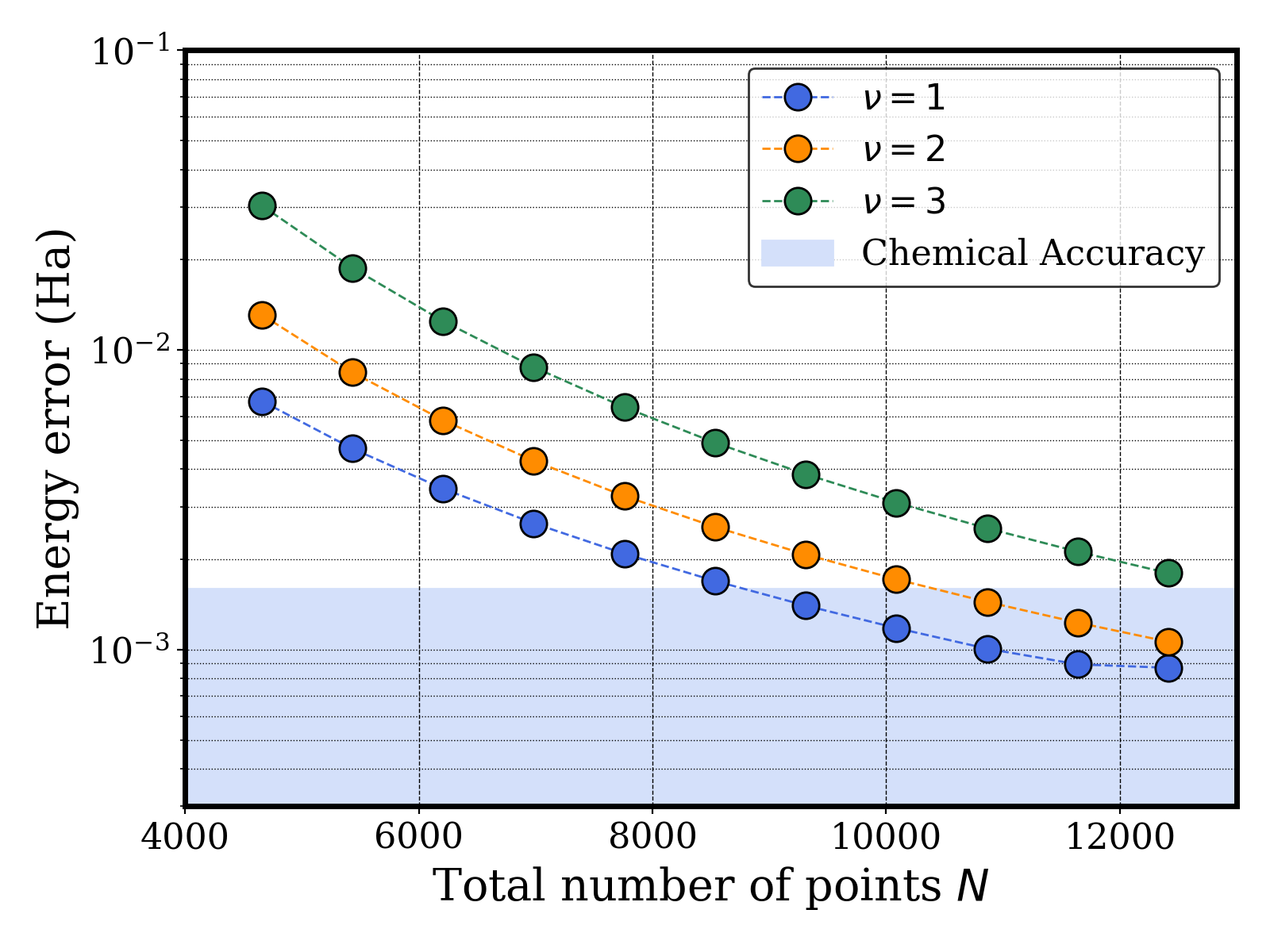}}
 \caption{%
  Convergence of the hydrogen ground-state energy for Becke–Lebedev grids with radial exponents \(\nu=1\) (blue), \(\nu=2\) (orange), and \(\nu=3\) (green). The shaded band marks chemical accuracy (1 mHa). Panels compare the non-transcorrelated (a) and transcorrelated Hamiltonians with \(\mu_{ne}=3\) (b) and \(\mu_{ne}=1\) (c).}
 \label{SI-fig:h_conv}
\end{figure*}

Figure \ref{SI-fig:e_cv} compares the energy convergence of the transcorrelated helium ground state on Gauss–Legendre and Lebedev angular grids. The Lebedev grid \CF{achieves the same level of accuracy} several shells earlier than Gauss–Legendre at the same angular resolution, highlighting its better point distribution. As typical in grid‐based approaches, the computed energy approaches the full-CI complete-basis-set limit “from below,” with smaller grids underestimating the true energy and each refinement incrementally raising the value toward the FCI-CBS limit. \CF{We also report the corresponding non-transcorrelated (NTC) simulations as a baseline. Overall, the TC results show lower errors than their NTC counterparts, confirming that the transcorrelated treatment improves the spatial description of the wave function.}

\begin{figure}[h]
 \centering
 \begin{minipage}{0.45\linewidth}
  \includegraphics[width=\linewidth]{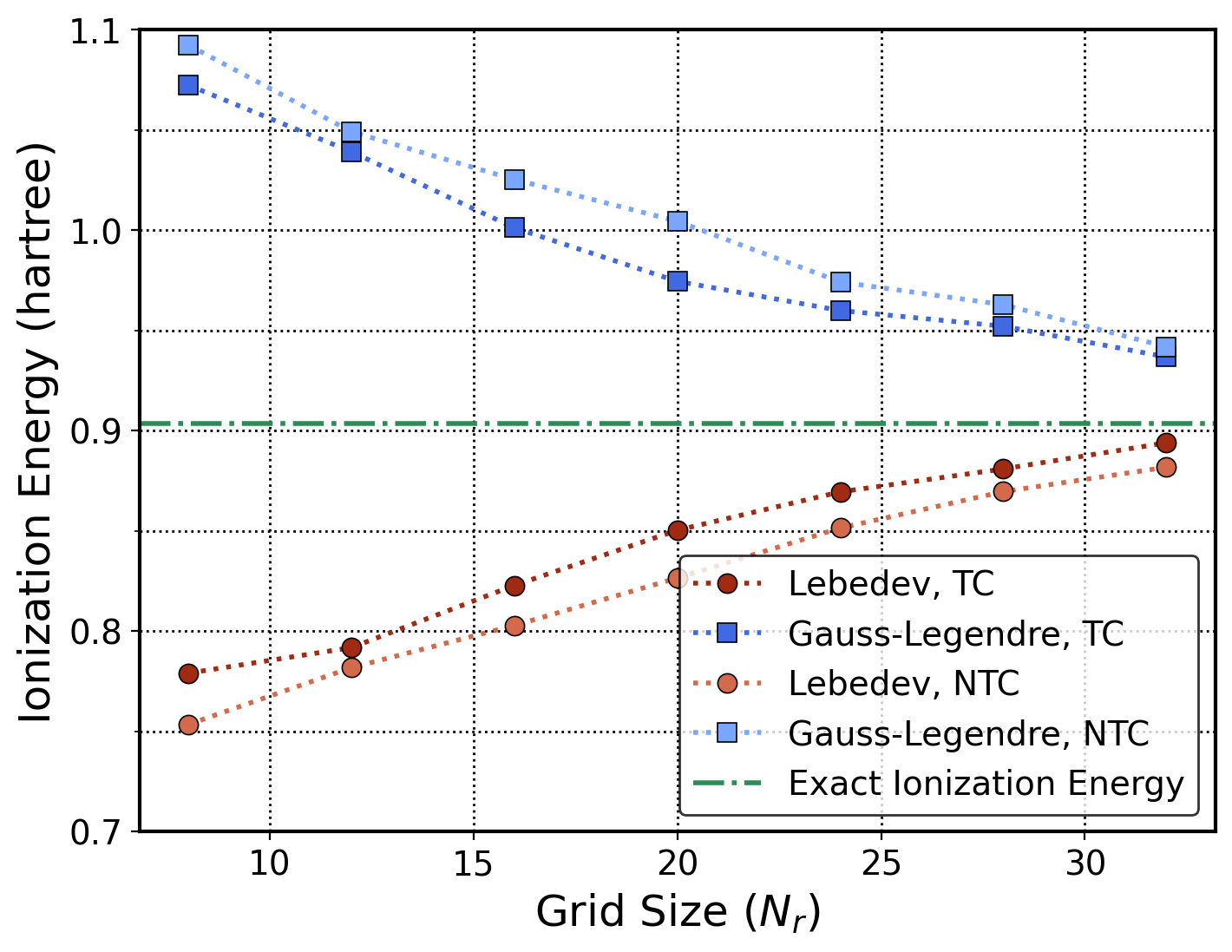}
 \end{minipage}%
 \hfill
 \begin{minipage}{0.49\linewidth}
\caption{\CF{Convergence of the helium first ionization energy with radial grid size for Gauss--Legendre and Lebedev quadratures, comparing transcorrelated (TC) and non-transcorrelated (NTC) simulations (128 angular points each, \(\mu_{ee}=2\), \(\mu_{ne}=1\) for TC). The dash-dotted line marks the complete-basis-set full-CI reference ionization energy of He.}}
  \label{SI-fig:e_cv}
 \end{minipage}
\end{figure}

\section{Block Encoding Procedure} \label{SI-LCUB}
\begin{definition}[Linear Combination of Unitaries] \label{SI-one-norm}
  Given the matrix representation $\mathbf{H}\in \mathbb{C}^{N \times N}$ of the Hamiltonian (or any other operator), which can be decomposed into a linear combination of unitaries (LCU), \textit{ie} for $L\le N^{2}$ and 
  \begin{equation}
    \mathbf{H}=\sum_{l=0}^{L-1}a_{l}\mathbf{U}_l
  \end{equation}
  where $a_l$ are some complex coefficients and $\mathbf{U}_l$ are $N \times N$ unitary operators, we define the one-norm of the LCU of $\mathbf{H}$ as
  \begin{equation}
  \lambda = \sum_{l=0}^{L-1}|a_{l}|  \label{SI-LCUnorm}
\end{equation}
\end{definition}

\begin{definition}[Block Encoding of Square Matrices]
  We refer to $ \mathbf{U}_H$ as a $(\lambda, \kappa, \epsilon)$-block encoding of $ \mathbf{H} \in \mathbb{C}^{N \times N}$ with $ \lambda > 0 $ and $ \epsilon > 0 $ if we have
\begin{equation}
\left\| \frac{\mathbf{H}}{\lambda} - \left( \langle 0|^{\otimes \kappa} \otimes \mathbf{I}_N \right) \mathbf{U}_H \left( |0\rangle^{\otimes \kappa} \otimes \mathbf{I}_N \right) \right\|_2 \leq \epsilon
\end{equation}
\end{definition}$\mathbf{I}_N$ is the $N$ by $N$ identity matrix. Sandwiching $\mathbf{U}_{H} \in \mathbb{C}^{2^{\kappa}N\ \times \ 2^\kappa N}$ between $|0\rangle^{\otimes \kappa} \otimes \mathbf{I}_N$ and its bra counterpart is just a mathematical way of saying that we extract the upper left $N$ by $N$ block of $\mathbf{U}_{H}$. When the block-encoding is exact, \textit{ie} $\epsilon=0$, $\mathbf{U}_H$ becomes a $(\lambda, \kappa)$ block encoding and it takes the form
\begin{equation}
  \mathbf{U}_H = \begin{pmatrix} \mathbf{H}/\lambda & * \\ * & * \end{pmatrix} \label{SI-block}
\end{equation}
where the blocks denoted by '$*$' do not contain useful information but ensure unitarity. A block encoding $\mathbf{U}_H$ of $\mathbf{H}$ can be obtained by combining the operators \texttt{PREP} (preparation) and \texttt{SELECT} operations defined as~\cite{low2019hamiltonian}:
\begin{equation}
\texttt{PREP} \ket{0}^{\otimes \log L} = \sum_{l=0}^{L-1} \sqrt{\frac{a_l}{\lambda}} \ket{l}, \space\texttt{SELECT} = \sum_{l=0}^{L-1} \ket{l} \bra{l} \otimes \mathbf{U}_l
\end{equation}
and performing fixed-point oblivious amplification on the $\ket{0}^{\otimes \kappa}$ state~\cite{Berry_2018,Yan2022} after constructing the gate
\begin{equation}
  (\texttt{PREP}^\dagger \otimes \mathbf{I}_N) \cdot \texttt{SELECT} \cdot (\texttt{PREP}\otimes \mathbf{I}_N)
\end{equation}Block-encoding then enables the construction of a unitary operator 
\begin{equation}
\mathbf{Q} = \mathbf{U}_{H}\begin{pmatrix} \mathbf{I}_{N\times N} & \mathbf{0}_{N\times(2^{\kappa}-N)}\\ \mathbf{0}_{(2^{\kappa}-N)\times N} & -\mathbf{I}_{(2^{\kappa}-N)\times(2^{\kappa}-N)} \end{pmatrix}
\end{equation} called the \textit{qubitized walk-operator}~\cite{sünderhauf2023generalizedquantumsingularvalue}. When the block-encoding is exact, $\mathbf{Q}$ has eigenvalues equal to $e^{\pm i\cos^{-1}(E_k/\lambda)}$, where $E_k$ are the eigenvalues of $\mathbf{H}$.

\section{Quantum Chebyshev Phase Estimation} \label{SI-QCPE}
In this section, we summarize the method proposed by Low and Su~\cite{Low_2024} to estimate the eigenvalues of our non-Hermitian Hamiltonian. We recall that this loss of Hermiticity is due to the additional non-Hermitian differential operators \eqref{nh1} and \eqref{nh2} that arise from the transcorrelation transformation and our inability to perform a symmetrization procedure similar to what was done in section \ref{finitevolumescheme}. This quantum algorithm is only applicable to square matrices with real eigenvalues, which we conjecture is the case for our Hamiltonian due to numerical experiments with exact diagonalization. The method can be extended to eigenvalues that are contained in an area of the complex plane, but for simplicity we will only consider square matrices with real eigenvalues in the interval $[-1;1]$.
\subsection{Overview}
As the transcorrelated Hamiltonian matrix $\mathbf{\tilde{H}}\in \mathbb{R}^{N^{\eta}\times N^{\eta}}$ is a non-Hermitian matrix, finding its eigenvalues with a quantum algorithm can no longer be accessed through the usual qubitized QPE procedure. Hence, one needs to consider a larger class of QEVE algorithms. Given an $\upsilon$-dimensional register, a state $\ket{\psi} \in \mathbb{C}^{N^{\eta}}$ that is prepared close to the ground state of $\mathbf{\tilde{H}}$, and a subnormalization constant $\alpha_{\mathbf{\tilde{H}}}$ of $\mathbf{\tilde{H}}$, previous approaches were based on generating (up to a normalization constant) the state~\cite{shao2020computingeigenvaluesdiagonalizablematrices,shao2021solvinggeneralizedeigenvalueproblems} 
\begin{equation}
    \sum_{\ell=0}^{\upsilon-1}\ket{\ell}e^{2\pi i\ell \frac{\mathbf{\tilde{H}}}{\alpha_{\tilde{H}}}} \ket{\psi}
\end{equation} and measuring the eigenvalue from the phase that is kick-backed onto the ancilla qubits. However, this algorithm suffers from a suboptimal $\mathcal{O}(\epsilon^{-1} \cdot\text{polylog}(\epsilon^{-1}) )$ query complexity dependence on $\epsilon$. To resolve this,~\cite{Low_2024} introduced an alternative QEVE algorithm based on generating \textit{Chebyshev history states}, and whose query complexity regains the $\mathcal{O}(1/\epsilon)$ of QPE. We summarize this \textit{Quantum Chebyshev Phase Estimation} (QCPE) algorithm above and we state the main theorem here. 
\begin{theorem}[Quantum Chebyshev Phase Estimation~\cite{Low_2024}] \label{SI-main}
Let $\mathbf{H}$ be a square matrix with only real eigenvalues such that $\mathbf{H}/\alpha_H$ is a block encoded by $O_H$ with some normalization factor $\alpha_H \ge 2||\mathbf{H}||_2$. Suppose that
oracle $O_{\psi}\ket{0}=\ket{\psi}$ prepares an initial state within distance $\big{|}\big{|}\ket{\psi}-\ket{\psi_0}\big{|}\big{|}=\mathcal{O}(\frac{\epsilon}{\alpha_H \alpha_U})$ from an eigenstate $\ket{\psi_0}$ such that $\mathbf{H}\ket{\psi_0}=E_{0}\ket{\psi_0}$, where
\begin{equation}
    \alpha_{U} \ge \max_{\ell=0,1\dots \upsilon-1} ||U_{\ell}(\frac{\mathbf{H}}{\alpha_H})||
  \end{equation}
is an upper bound on the Chebyshev polynomials of the second kind $U_\ell(x)$ \eqref{SI-cheb2} with $\upsilon=\mathcal{O}(\frac{\alpha_H}{\epsilon})$. Then, $E_0$ can be estimated with accuracy $\epsilon$ and probability $1-p_f$ using 
  \begin{equation}
    \mathcal{O}\Big{(}\frac{\alpha_{H}}{\epsilon}\alpha_{U}\log\big{(}\frac{1}{p_f}\big{)} \Big{)}
  \end{equation}
  queries to controlled-$O_H$, controlled-$O_\psi$ and their inverses.
\end{theorem}

\subsection{Preliminaries on Chebyshev Polynomials}

\begin{definition}[Chebyshev Polynomials of the First and Second Kinds~\cite{garfken67:math}] \hfill
  \begin{itemize}
    \item The $\ell$th Chebyshev polynomial of the first kind $T_{\ell}:[-1;1]\rightarrow[-1;1]$ is defined as: \begin{equation}
      T_{\ell}(x)=\cos\big{(}\ell\arccos(x)\big{)} \label{SI-cheb1}
    \end{equation}
    with generating function
    \begin{equation}
    \sum_{\ell=0}^{\infty} T_{\ell}(x)y^{\ell}=\frac{1-yx}{1+y^2-2yx} \label{SI-gen1}
    \end{equation}
    It is also convenient to define the rescaled Chebyshev polynomials of the first kind:
    \begin{equation}
    \tilde{T}_{\ell}(x)=
      \begin{cases}
    \frac{1}{2}T_{0}(x) & \text{if }\, \ell=0\\
    T_\ell(x) & \text{if }\, \ell\ge 1 \label{SI-cheb1r}
      \end{cases}
    \end{equation}
    with generating function
    \begin{equation}
    \sum_{\ell=0}^{\infty} \tilde{T}_{\ell}(x)y^{\ell}=\frac{1-yx}{2(1+y^2-2yx)} \label{SI-gen}
    \end{equation}

    \item The $\ell$th Chebyshev polynomial of the second kind $U_{\ell}:[-1;1]\rightarrow\mathbb{R}$ is defined as:\begin{equation}
      U_{\ell}(x)=\frac{\sin\big{(}(\ell+1)\arccos(x)\big{)}}{\sin\big{(}\arccos(x)\big{)}} \label{SI-cheb2}
    \end{equation}
    with generating function
    \begin{equation}
    \sum_{\ell=0}^{\infty} U_{\ell}(x)y^{\ell}=\frac{1}{1+y^2-2yx} \label{SI-gen2}
    \end{equation}
  
  \end{itemize}
\end{definition}

\subsection{The Algorithm}

\begin{definition}[Chebyshev History State] \label{SI-CHEBY}
  Let $\mathbf{H}$ be a square matrix with real eigenvalues and subnormalization constant $\alpha_{H} > \lVert \mathbf{H} \rVert _{2}$. Consider an $\upsilon$-dimensional register and an arbitrary state $\ket{\psi}$. The total system is said to form a Chebyshev history state of $\mathbf{H}$ if it can be written as \begin{equation}
    \sum_{\ell=0}^{\upsilon-1}\ket{\ell} \otimes {T}_{\ell}\big{(}\frac{\mathbf{H}}{\alpha_{H}}\big{)}\ket{\psi} \label{SI-CHS}
  \end{equation}
  up to a normalization constant, where $T_\ell$ is the $\ell$th Chebyshev polynomial of the first kind as defined in \eqref{SI-cheb1}.
\end{definition}
The reason for the name \textit{history state} is the fact that the $\upsilon$-dimensional register $\ket{\ell}$ serves as a counter that indicates the order of the Chebyshev polynomial of $\mathbf{H}/\alpha_H$ that is applied on the state $\ket{\psi}$ in the sum.
Suppose now that $\ket{\psi}$ is the eigenvector of $\mathbf{H}$ corresponding to its ground state, \textit{ie} $\ket{\psi}=\ket{\psi_0}$ such that $\mathbf{H}\ket{\psi_0}=E_{0}\ket{\psi_0}$. Then, the history state takes the form
\begin{equation}
  {\sum_{\ell=0}^{\upsilon-1} \ket{\ell} \otimes {T}_{\ell}\big{(}\frac{\mathbf{H}}{\alpha_{H}}\big{)}\ket{\psi_0}}= \sum_{\ell=0}^{\upsilon-1} {T}_{\ell}\big{(}\frac{E_0}{\alpha_{H}}\big{)} \ket{\ell} \otimes \ket{\psi_0} = \sum_{\ell=0}^{\upsilon-1} \cos{(2\pi \ell \phi)} \ket{\ell} \otimes \ket{\psi_0} 
\end{equation}
where $\phi = \frac{1}{2\pi}\arccos{\frac{E_0}{\alpha_H}}$. We now neglect the $\ket{\psi_{0}}$ register, whose purpose was to phase-kickback the cosine term to the $\upsilon$-dimensional register, and apply a quantum Fourier transform:
\begin{equation}
  QFT \sum_{\ell=0}^{\upsilon-1} \cos{(2\pi \ell \phi)} \ket{\ell} = \frac{1}{2\sqrt{\upsilon}} \sum_{\ell=0}^{\upsilon-1} \sum_{\ell'=0}^{\upsilon-1} \Big{(} e^{2\pi i\ell'(\phi-\frac{\ell}{\upsilon})} + e^{-2\pi i\ell'(\phi+\frac{\ell}{\upsilon})}   \Big{)} \ket{\ell} 
\end{equation}
If we scale $\alpha_H> 2||\mathbf{H}||_{2}$, the authors show that, by measuring $\ell$ in the computational basis, we obtain an approximate of $\upsilon\phi$ with a success probability greater than $1/2$. More precisely, starting from a Chebyshev history state \eqref{SI-CHS}, performing a QFT, and measuring in the computational basis gives a value $\ell$ satisfying
\begin{equation}
  |cmod_{1}(\frac{\ell}{\upsilon} \pm \phi)|<\frac{\upsilon'}{\upsilon}
\end{equation}
with probability at least $0.566$ for $\upsilon' \ge 5$. Here, $cmod$ refers to the centered modulus function defined as
\begin{equation}
  cmod_{q}(x)=x-q\lfloor \frac{x+\frac{q}{2}}{q} \rfloor
\end{equation}
The success probability can then be boosted to at least $1-p_f$ by repeating the procedure $\mathcal{O}(\log(1/p_f))$ times. However, one must first be able to generate Chebyshev history states.

\begin{lemma}
  Let $\mathbf{H}$ be a square matrix with real eigenvalues and subnormalization constant $\alpha_{H} > ||\mathbf{H}||_{2}$ and let $\mathbf{L}_\upsilon=\sum_{\ell=0}^{\upsilon-2}\ket{\ell+1}\bra{\ell}$ be the $\upsilon \times \upsilon$ lower shift operator. The matrix version of the operator of the generating function for rescaled Chebyshev polynomials \eqref{SI-gen} \begin{equation}
    G\big{(}\frac{\mathbf{H}}{\alpha_H}\big{)} = \sum_{\ell=0}^{\upsilon-1}\mathbf{L}_\upsilon^\ell \otimes \tilde{T}_{\ell}\big{(}\frac{\mathbf{H}}{\alpha_H}\big{)}
  \end{equation}
  applied on the $\ket{0}\otimes \ket{\psi}$ state gives the rescaled Chebyshev history state
  \begin{equation}
    \sum_{\ell=1}^{\upsilon-1} \ket{\ell} \otimes \tilde{T}_{\ell}\big{(}\frac{\mathbf{H}}{\alpha_{H}}\big{)}\ket{\psi} \label{SI-RCHS-eq}
  \end{equation} \label{SI-RCHS}
\end{lemma}

\begin{proof}
  It is a straightforward application of the operator $G$, noticing that elevating the $ \upsilon \times \upsilon$ lower shift operator to the $\ell$th power gives \begin{equation}
    \mathbf{L}_\upsilon^{\ell}=\sum_{\ell'=0}^{\upsilon-1-\ell} \ket{\ell+\ell'}\bra{\ell'}
  \end{equation}
\end{proof}
In order to use the Chebyshev quantum phase estimation algorithm, we thus have to apply this operator to the zeroed $\upsilon$-dimensional register. To implement it as a gate on a quantum computer, we first define a padding of the matrix $\mathbf{H}/\alpha_{H}$:
\begin{equation}
\begin{split}
  Pad(\frac{\mathbf{H}}{\alpha_H})&=\mathbf{I}_{\upsilon}\otimes\mathbf{I} + \mathbf{L}_{\upsilon}^2 \otimes \mathbf{I} -2\mathbf{L}_{\upsilon}\otimes \frac{\mathbf{H}}{\alpha_H}\\
  &=\begin{pmatrix}
   \mathbf{I} & 0 & \dots & \dots & \dots &0 \\
   -2\mathbf{H}/\alpha_H & \mathbf{I} & \ddots & \ddots & \ddots & \vdots \\
   \mathbf{I} & -2\mathbf{H}/\alpha_H & \mathbf{I} & \ddots & \ddots & \vdots \\
   0 & \mathbf{I} & -2\mathbf{H}/\alpha_H & \mathbf{I} & \ddots & \vdots \\
   \vdots & \ddots & \ddots & \ddots & \ddots & \vdots \\
   0 & \dots & 0 & \mathbf{I} & -2\mathbf{H}/\alpha_H & \mathbf{I}
  \end{pmatrix}
\end{split} \label{SI-pad}
\end{equation}
and use the following lemma.

\begin{lemma}
  \begin{equation}
    G\big{(}\frac{\mathbf{H}}{\alpha_H}\big{)} = Pad\big{(}\frac{\mathbf{H}}{\alpha_H}\big{)}^{-1} . \Big{(}\frac{\mathbf{I}_\upsilon -\mathbf{L}_\upsilon^2}{2}\otimes \mathbf{I}\Big{)}
  \end{equation}
\end{lemma}

\begin{proof}
  Inverting the right-hand side of \eqref{SI-pad} and multiplying by $\frac{\mathbf{I}_\upsilon -\mathbf{L}_\upsilon^2}{2}\otimes \mathbf{I}$, we obtain the form of the generating function for the rescaled Chebyshev polynomials of \eqref{SI-gen} with $x=\mathbf{I}_\upsilon \otimes \frac{\mathbf{H}}{\alpha_H}$ and $y=\mathbf{L}_\upsilon \otimes \mathbf{I}$.
\end{proof}
Since finding the inverse of a matrix, let alone on a quantum computer, is a tedious task, we will get our Chebyshev history state by solving the linear system
\begin{equation}
  Pad(\frac{\mathbf{H}}{\alpha_H}) \ket{\Phi}\otimes\ket{\psi}=\frac{\ket{0}-\ket{2}}{2}\otimes \ket{\psi} \label{SI-linsys}
\end{equation}
for a state $\ket{\Phi}\in\mathcal{H}^{\upsilon}$. The right-hand side of the equation results from applying the operator $\frac{\mathbf{I}_\upsilon -\mathbf{L}_\upsilon^2}{2}\otimes \mathbf{I}$ to the $\ket{0}\otimes\ket{\psi}$ state. Many quantum linear system solvers have been proposed, but the authors propose to use the one described in~\cite{costa2021optimalscalingquantumlinear}. The task of block-encoding the gate $Pad(\mathbf{H}/\alpha_H)$ and the lower shift matrix is one that will not be discussed in this work, as it has been extensively elaborated on in~\cite{Low_2024}.

\begin{figure}[htp]
\centering
\begin{minipage}{0.95\linewidth}
\hrule
\vspace{0.6em}

\textbf{Algorithm: Quantum Chebyshev Phase Estimation}
\label{SI-CHG}

\vspace{0.4em}

\begin{enumerate}
    \item Construct a block encoding of the target matrix
    $\mathbf{H}$ and the lower-shift matrix $\mathbf{L}_{\upsilon}$.

    \item Use $\mathbf{L}_{\upsilon}$ and $\mathbf{H}$ to construct the
    $\operatorname{Pad}(\mathbf{H}/\alpha_H)$ gate according to
    Eq.~\eqref{SI-pad}.

    \item Prepare the $\upsilon$-dimensional register in the state
    \[
        \frac{\ket{0}-\ket{2}}{\sqrt{2}}.
    \]

    \item Invoke the quantum linear-system solver to solve
    Eq.~\eqref{SI-linsys} and prepare $\ket{\Phi}$.

    \item Perform the quantum Fourier transform and measure in the
    computational basis.

    \item Boost the success probability using median amplification.
\end{enumerate}

\vspace{0.3em}
\hrule
\end{minipage}
\caption{Quantum Chebyshev Phase Estimation algorithm.}
\label{SI-fig:CHG}
\end{figure}
The general theorem for the Chebyshev phase estimation \ref{SI-main} takes into account (i) the error from the quantum linear system solver, (ii) the distance between the prepared state $\ket{\psi}$ and the true ground state eigenstate $\ket{\psi_0}$ and (iii) the fact that we are approximating a rescaled Chebyshev history state \eqref{SI-RCHS} to a normal one \eqref{SI-CHS}.
Finally, for preparing an initial state that is close to the true ground state, the authors introduce an algorithm for preparing the ground state of a non-Hermitian matrix using \textit{Quantum EigenValue Transformation} (QEVT). It mainly consists of forming a more general Chebyshev state than the one in \eqref{SI-CHS}, which requires the construction of an extended $Pad(A)$ matrix.

\bibliographystyle{unsrt}
\bibliography{refs,proof_additions}

\end{document}